\documentclass[preprint,authoryear, 12pt]{elsarticle}
\usepackage{amsmath,amssymb,mathrsfs}
\numberwithin{equation}{section}
\usepackage{graphicx}
\usepackage{subfigure}

\usepackage{mathtools}
\usepackage{amsfonts}
\usepackage{cases}
\usepackage{color}

\usepackage{bbm}
\usepackage{enumerate}

\usepackage[thmmarks]{ntheorem}
{
\theoremstyle{nonumberplain}
\theoremheaderfont{\bfseries}
\theorembodyfont{\normalfont}
\theoremsymbol{\mbox{$\Box$}}
\newtheorem{proof}{Proof}
}
\newtheorem{theorem}{Theorem}[section]
\newtheorem{lemma}{Lemma}[section]

\newtheorem{remark}{Remark}[section]
\newtheorem{proposition}{Proposition}[section]

\newtheorem{assumption}{Assumption}[section]
\usepackage{titlesec}

\usepackage{geometry}
\journal{arXiv}

\begin{document}
\begin{frontmatter}
\title{Consumption-investment optimization with Epstein-Zin utility in unbounded non-Markovian markets}
\author[math1]{Zixin Feng}
\ead{zixin.feng@whu.edu.cn}
\author[math2]{Dejian Tian\corref{correspondingauthor}}
\ead{djtian@cumt.edu.cn}
\author[math3]{Harry Zheng}
\ead{h.zheng@imperial.ac.uk}
\address[math1]{School of Mathematics and Statistics, Wuhan University, Wuhan, P.R. China}
\address[math2]{School of Mathematics, China University of Mining and Technology, Xuzhou, P.R. China}
\address[math3]{Department of Mathematics, Imperial College, London SW7 2BZ, UK}

\cortext[correspondingauthor]{Corresponding author}

\begin{abstract}
The paper investigates the consumption-investment problem for an investor with Epstein-Zin utility in an incomplete market. A non-Markovian environment with unbounded parameters is considered, which is more realistic in practical financial scenarios compared to the Markovian setting. The optimal consumption and investment strategies are derived using the martingale optimal principle and quadratic backward stochastic differential equations (BSDEs) whose solutions admit some exponential moment. This integrability property plays a crucial role in establishing a key martingale argument. In addition, the paper also examines the associated dual problem and several models within the specified parameter framework.
\end{abstract}

\end{frontmatter}
\section{Introduction}
In classical asset pricing models, an agent is assumed to have a time-additive von Neumann-Morgenstern utility, with two well-known types being constant absolute risk aversion (CARA) and constant relative risk aversion (CRRA). Since the landmark work of \cite{merton1971optimum} utilizing the dynamic programming approach, there has been a substantial amount of literature focusing on the optimal consumption and investment problem with these utilities, including \cite{karatzas1987optimal} developing the martingale approach, \cite{kramkov1999asymptotic} providing the duality result, \cite{hu2005utility} introducing the technique of backward stochastic differential equations (BSDEs) and \cite{horst2014forward} further proposing a new approach in terms of a fully coupled system of forward backward stochastic differential equations (FBSDEs).

However, widely used time-additive utilities inattentively restrict an artificial relationship between risk aversion $\gamma$ and the elasticity of intertemporal substitution (EIS) $\psi$, leading to abundant asset pricing anomalies such as the risk-free rate puzzle, the equity premium puzzle, etc. To untie risk aversion from EIS, \emph{Epstein-Zin} type recursive utility and its continuous-time counterpart, stochastic differential utility, were proposed by \cite{epstein1989substitution} and \cite{duffie1992stochastic} respectively, the latter of which could be defined rigorously using BSDEs. \cite{kraft2014stochastic} have formally established the connection between these two formulations. One can refer to \cite{schroder1999optimal,seiferling2015stochastic,xing2017consumption} for the existence and uniqueness results of Epstein-Zin stochastic differential utility (SDU) and its properties. Compared to time-additive utilities, recursive utilities and SDU provide a more general framework for addressing the aforementioned asset pricing anomalies, as demonstrated in \cite{bansal2004risks} and \cite{benzoni2011explaining}.

A significant amount of literature has focused on consumption and investment optimization with Epstein-Zin utility. For a non-Markovian setting, \cite{schroder1999optimal} developed the utility gradient approach with $\theta=\frac{1-\gamma}{1-1/\psi}>0$.  \cite{schroder2003optimal, schroder2005lifetime} extended it to include convex trading constraints under this parameter specification with unit EIS. \cite{el2001dynamic} stated a maximum principle for generalized recursive utilities by  FBSDEs in the presence of nonlinear wealth. For a Markovian financial market, \cite{kraft2013consumption,seiferling2015stochastic,kraft2017optimal} applied the dynamic programming principle to get the Hamilton-Jacobi-Bellman (HJB) equation and established verification theorems, allowing for many possible configurations of $\gamma$ and $\psi$. \cite{xing2017consumption} adopted the martingale optimal principle combined with BSDEs for $\gamma,\psi>1$.  \cite{matoussi2018convex} established a dual inequality between the primal and dual problems when $\gamma\psi\geq1,\psi>1$ or $\gamma\psi\leq1,\psi<1$. Other papers include \cite{herdegen2022infinite}  for the infinite time  Epstein-Zin utility model, \cite{aurand2023epstein,pu2024consumption} for applications, etc.


This paper investigates the optimal consumption and investment problem under Epstein-Zin utility in a non-Markovian market with unbounded parameters and $\gamma,\psi>1$, which is beyond \cite{matoussi2018convex} that requires bounded interest rate and market price of risk. {
The main contributions of the paper are the following three aspects. First, we extend the analysis to a non-Markovian model with unbounded parameters, especially unbounded market price of risk. {
\cite{schroder1999optimal,schroder2003optimal} developed the utility gradient approach in a non-Markovian market with bounded market parameters. The former studied the bounded market price of risk and the latter the bounded interest rate. \cite{matoussi2018convex} proposed a dual problem in a non-Markovian market with bounded market price of risk. We relax the assumption of bounded parameters and instead consider exponentially integrable parameters. We use  the martingale optimal principle to establish a non-Markovian BSDE with unbounded solutions; such a methodology was first proposed by \cite{hu2005utility}, then used by \cite{cheridito2011optimal} for time-additive utilities and  later by \cite{xing2017consumption} for the Epstein-Zin framework. 

Second, we obtain the martingale property with the martingale optimal principle and the dual equality with the exponential integrability of model parameters. In a non-Markovian market, we cannot rely on HJB equations from \cite{kraft2013consumption,kraft2017optimal} nor Lyapunov functions from \cite{xing2017consumption}, which are only applicable in a Markovian setting. {
The crucial difficulty with the martingale optimal principle is to verify that a certain exponential local martingale is a true martingale for the candidate optimal strategy. {
\cite{matoussi2018convex} considered a non-Markovian market with bounded interest rate and market price of risk, which ensures the martingale property automatically. We first derive some precise estimates on the solution of the BSDE, especially the exponential integrability property, and then use de La Vall$\rm\acute{e}$e Poussin's lemma with the relative entropy method to directly characterize the martingale property from the BSDE itself, without relying on specific market structure. To our best knowledge, there are few papers discussing the convex duality in an incomplete market with stochastic differential utility except \cite{matoussi2018convex} for a non-Markovian model with bounded parameters. We find that the same duality (in)equality can be constructed for a non-Markovian model with unbounded parameters that have finite exponential moment of some order and, following some verified conclusions from \cite{matoussi2018convex}, obtain the optimal state price density and the associated class (D) property. 


Third, our approach is  applicable to some well-known models with non-Markovian structure and possibly unbounded parameters, such as the Heston model, linear diffusion model and CIR model. Several papers have investigated these models in a Markovian market or a non-Markovian market with bounded parameters, see \cite{hu2005utility,cheridito2011optimal,gu2016dual,xing2017consumption,matoussi2018convex,feng2023optimal}. Under our assumption that  market parameters have an exponential moment of a certain order, BSDE techniques contribute to identify the optimal strategy of consumption and investment  in these models. 

Motivated by \cite{hu2018exponential} and \cite{delbaen2011uniqueness}, both \cite{hu2024utilityv4, hu2024utility} and we studied Epstein-Zin optimization in unbounded markets, independently of each other. The ideas and methodologies are largely the same. The main differences are threefold: First, $r=0$ in \cite{hu2024utility} while $r$ is exponentially integrable in our case. When $r$ is unbounded, the solution to the characterizing BSDE (\ref{mar}) becomes unbounded--a scenario not addressed in \cite{hu2024utility}. We established the martingale property for the key local martingale with unbounded $r$. Second, $0<\gamma\ne 1, \psi>1$ in \cite{hu2024utility} while  $\gamma>1,\psi>1$ in our case  to ensure  the parameter range for the dual problem satisfies $\gamma\psi\geq1,\psi>1$, which also makes our integrability conditions for the admissible class under convex duality further relaxed. Third, \cite{hu2024utility} investigated problems with convex closed set constraints while we did not. The analysis of the growth of the generator of the characterizing BSDE becomes more  complex, but its growth nature remains fundamentally unchanged. Apart from Epstein-Zin optimization, \cite{hu2024utilityv4, hu2024utility} also discussed some other topics in utility maximization, an excellent research memoir and reference book.

The remainder of the paper is organized as follows. Section 2 introduces the Epstein-Zin utility in a non-Markovian market and the associated consumption-investment problem. Section 3 discusses the main results, including the martingale optimal principle, the verification theorem, and the convex duality under our specified parameter conditions.  Section 4 presents three examples in non-Markovian markets, including Heston, linear diffusion, and CIR models. Section 5 contains  the proofs of main results (Proposition \ref{Y}, Theorems \ref{main} and \ref{dual}). Section 6 concludes the paper.

\section{The consumption-investment optimization under Epstein-Zin utility}

In this section, we formulate the consumption-investment optimization problem under Epstein-Zin utility in a non-Markovian financial market. 


\subsection{The non-Markovian model setup}
Given a time horizon $T<+\infty$. Let $\left(\Omega,\mathscr{F},\left(\mathscr{F}_t\right)_{0\leq t\leq T},\mathbb{P}\right)$ be a filtered probability space in which $\left(\mathscr{F}_t\right)_{0\leq t\leq T}$ is the natural filtration generated by a $n$-dimensional standard Brownian motion $W$, satisfying the usual hypotheses, completeness and right-continuity. 

Consider a financial market model comprising of a risk-free asset $S^0$ and risky assets $S=(S^1,\cdots, S^d)$, which satisfy the dynamics
\begin{align*}
&dS_t^0=r_tS_t^0dt,\\
&dS_t=\operatorname{diag}(S_t)\left[\left(r_t\mathbf{1}_d+\mu_t\right)dt+\sigma_tdW_t\right],
\end{align*}
where diag($S$) is a diagonal matrix with the elements of $S$ on the diagonal and $\mathbf{1}_d$ is a $d$-dimensional vector with each entry 1. The model coefficients $r$, $\mu$, $\sigma$ are $\mathbb{R}$-valued, $\mathbb{R}^d$-valued, $\mathbb{R}^{d\times n}$-valued predictable stochastic processes, called risk-free rate, excess return and volatility processes respectively. We assume $d\leq n$, so the market can be incomplete.

We need to impose the following assumptions for the main results.  
\begin{assumption}\label{1}
For some $p>1$ and $q>\max\left\{\frac{p(\gamma-1)(\gamma+2)}{\gamma(1+2(p-1)\gamma)}, 1\right\}$,
\begin{align}\label{11}
\mathbf{E}\left[\exp\left(2p\gamma\int_0^Tr_s^- ds\right)+\exp\left(q(\gamma-1)\int_0^Tr_s^+ ds\right)+\exp\left(q\int_0^T\mu'_s\Sigma_s^{-1}\mu_sds\right)\right]<+\infty,
\end{align}
where $\Sigma:=\sigma\sigma'$ is positive definite and the relative risk aversion $\gamma>1$. 
\end{assumption}

\begin{assumption}\label{2}
$r^-$ is bounded.
\end{assumption}

\begin{remark}\label{1interpretation}
Assumption \ref{1} implies that $\mathbf{E}\left[\exp\left(\frac{1}{2}\int_0^T\mu'_s\Sigma_s^{-1}\mu_sds\right)\right]<+\infty$, which indicates the market has no arbitrage opportunity by Novikov's condition.  Assumption \ref{2} is reasonable from the viewpoint of finance. 
There is a large literature on studying the Epstein-Zin utility model with different conditions on the parameters. For example, for Markovian models, $r,\mu,\sigma$ are all bounded in \cite{kraft2013consumption,kraft2017optimal}, $r+\frac{\mu'\Sigma^{-1}\mu}{2\gamma}$ is bounded from below in \cite{xing2017consumption}; for non-Markovian models, $r$ is bounded in \cite{schroder2003optimal},  $r$ and $\mu'\Sigma^{-1}\mu$ are both bounded or Markovian in \cite{matoussi2018convex}. 
\end{remark}

Assumption~\ref{1} looks technical. We thank the referees for suggesting a simplified sufficient condition, Assumption~\ref{3},  that implies  Assumptions~\ref{1} and \ref{2} due to the fact that $\lim_{p\to\infty}\frac{p(\gamma-1)(\gamma+2)}{\gamma(1+2(p-1)\gamma)}={(\gamma-1)(\gamma+2)\over 2\gamma^2}<1$ and we may choose $p$ sufficiently large such that $\max\left\{\frac{p(\gamma-1)(\gamma+2)}{\gamma(1+2(p-1)\gamma)}, 1\right\}=1$.
\begin{assumption}\label{3}
$r^-$ is bounded and for some $q>1$,
\begin{align*}
\mathbf{E}\left[\exp\left(q(\gamma-1)\int_0^Tr_s^+ ds\right)+\exp\left(q\int_0^T\mu'_s\Sigma_s^{-1}\mu_sds\right)\right]<+\infty.
\end{align*}
\end{assumption}

Let $\mathcal{R}_+$ be the set of all nonnegative progressively measurable processes on $[0,T]\times\Omega$. For $c\in\mathcal{R}_+$,  $c_t$ represents the consumption rate at time $t$ when $t<T$ and a lump sum consumption when $t=T$. Let $\mathcal{R}^d$ denote the set of all predictable processes taking values in $\mathbb{R}^d$.  

An agent invests in this financial market by selecting a consumption process $c\in \mathcal{R}_+$ and an investment strategy $\pi\in \mathcal{R}^d$ such that the associated wealth process $\mathcal{W}^{c,\pi}$ is nonnegative. Given an initial wealth $\omega$, $\mathcal{W}^{c,\pi}$ is given by, for $t\in [0,T]$, 
\begin{align}\label{wealth_process}
d\mathcal{W}_t^{c,\pi}=\mathcal{W}_t^{c,\pi}\left(\left(r_t+\pi_t'\mu_t\right)dt+\pi_t'\sigma_t dW_t\right)-c_tdt,\quad\mathcal{W}_0^{c,\pi}=\omega.
\end{align}
For simplicity of notation, we write $\mathcal{W}$ instead of $\mathcal{W}^{c,\pi}$ in the rest of the paper. 

\subsection{Epstein-Zin preferences and the associated optimization problem}\label{Epstein-Zin preferences}
Consider an agent whose preference over $\mathcal{R}_+$-valued consumption streams is described by a continuous-time stochastic differential utility of \emph{Epstein-Zin} type. To this end, recall the relative risk aversion $\gamma>1$ and let $\delta>0$ and $\psi>1$ stand for the discounting rate and the EIS respectively. Given a bequest utility function $V(c)=\frac{c^{1-\gamma}}{1-\gamma}$, the \emph{Epstein-Zin} utility for a pair of strategy ($c,\pi$) over a time horizon $T$ is a semimartingale $V^{c,\pi}$ satisfying, for $t\in [0,T]$, 
\begin{align}\label{EZuc}
	V_t^{c,\pi}=\mathbf{E}\left[\int_t^T f(c_s, V_s^{c,\pi})ds+V(\mathcal{W}_T)\Big|\mathscr{F}_t\right],
\end{align}
where $f:[0,+\infty)\times(-\infty,0]\to\mathbb{R}$ represents the Epstein-Zin aggregator, defined by
\begin{align}\label{EZa}
f(c,v)=\frac{\delta c^{1-\frac{1}{\psi}}}{1-\frac{1}{\psi}}\left((1-\gamma)v\right)^{1-\frac{1}{\theta}}-\delta\theta v,\quad {\rm with\ } \theta:=\frac{1-\gamma}{1-\frac{1}{\psi}}<0. 
\end{align}
A pair of strategy $(c,\pi)$ is called \emph{admissible} if it belongs to
\begin{align}\label{adset}
\mathcal{A}=\left\{(c,\pi): V^{c,\pi} {\rm\ exists\ }, V^{c,\pi}<0, ~ ~ V^{c,\pi} ~and~\  \mathcal{W}^{1-\gamma}e^Y {\rm\ are\ of\ class\ (D)}\right\},
\end{align}
where $Y$ is defined by BSDE (\ref{mar}). There is a large amount of literature on the admissible class  for the Epstein-Zin optimization including sufficient conditions for the existence of Epstein–Zin utility, which in particular indicates that $\mathcal{A}\neq\emptyset$, see \cite{xing2017consumption, aurand2023epstein}.   The class $(\rm D)$ condition is first proposed in \cite{hu2005utility} for time-additive utility. The class (D) property of $\mathcal{W}^{1-\gamma}e^Y$ is to ensure that $V$ is well-defined. From an economic perspective, admissibility, particularly the class (D) property, ensures that consumption-investment strategies remain financially viable over time. 
The reader can refer to Proposition 2.1 in \cite{matoussi2018convex} for general discussion on admissible strategies. 

The Epstein-Zin utility  maximization problem is the following: 
\begin{align}\label{max}
V_0:=\sup_{(c,\pi)\in\mathcal{A}}V_0^{c,\pi}=\sup_{(c,\pi)\in\mathcal{A}}\mathbf{E}\left[\int_0^T f(c_s, V_s^{c,\pi})ds+V(\mathcal{W}_T)\right].
\end{align}


\section{Main results}
In this section, with the help of the elaborate and effective priori estimate technique of BSDE, Proposition \ref{Y} investigates a specific BSDE \eqref{mar} carefully, which is acquired by the martingale optimal principle. Then, the verification theorem is established for the candidate optimal strategy and the optimal value function (see Theorem \ref{main}) in a non-Markovian market environment using the results of Proposition \ref{Y}. Finally, some connections with Epstein-Zin's dual result are also discussed.   
\subsection{Consumption-investment optimization}\label{mop}

In accordance with the special formulation of Epstein-Zin utility, for any $(c,\pi)\in \mathcal{A}$,  we construct a utility process:
\begin{align}\label{upo}
G_t^{c,\pi}:=\frac{\mathcal{W}_t^{1-\gamma}}{1-\gamma}e^{Y_t}+\int_0^tf\left(c_s,\frac{\mathcal{W}_s^{1-\gamma}}{1-\gamma}e^{Y_s}\right)ds,\quad t\in[0,T],
\end{align}
where $\mathcal{W}$ is the wealth process in (\ref{wealth_process}) and $Y$ satisfies, for $t\in [0,T]$, 
\begin{align}\label{mar}
 Y_t=\int_t^TH(s,Y_s,Z_s)ds-\int_t^TZ_sdW_s,
\end{align}
where 
\begin{align}\label{generator-complete}
H(t,y,z)=&z\left(\frac{1}{2}I_n+\frac{1-\gamma}{2\gamma}\sigma_t'\Sigma_t^{-1}\sigma_t\right)z'
+\frac{1-\gamma}{\gamma}\mu_t'\Sigma_t^{-1}\sigma_tz'+\frac{\theta}{\psi}\delta^\psi e^{-\frac{\psi}{\theta}y}\notag\\
&+\frac{1-\gamma}{2\gamma}\mu_t'\Sigma_t^{-1}\mu_t+(1-\gamma) r_{t}-\delta \theta.
\end{align}

\begin{remark}
$H$ is chosen to ensure $G^{c,\pi}$ in  (\ref{upo})   is a local supermartingale for any control $(c,\pi)$ and a local martingale for the candidate optimal control $(c^*,\pi^*)$ under the parameter specification $\gamma, \psi>1$, see Proposition \ref{g}. Similar equations were first introduced by \cite{hu2005utility} for the time-additive utility maximization in an incomplete market, where they constructed BSDEs whose generators are quadratic in $z$ but don't contain $y$. The generator \eqref{generator-complete} has a quadratic term in $z$ and also an exponential term in $y$, which is the same as equation (2.13) in \cite{xing2017consumption}. 
\end{remark}

The next result plays an important role in proving the martingale property in Theorem \ref{main}. 
Compared with Proposition 2.9 in \cite{xing2017consumption}, where the author imposed the condition that $r+\frac{\mu'\Sigma^{-1}\mu}{2\gamma}$ is bounded from below, our integrability condition is stronger as we need it to establish better prior estimates of the solutions that would help us to bypass Lyapunov's arguments for non-Markovian models.

 
\begin{proposition}\label{Y}
Suppose that $\gamma, \psi>1$ and Assumption \ref{1} holds. Then BSDE (\ref{mar}) admits a unique solution $(Y,Z)$ such that 
\begin{align}\label{exin}
\mathbf{E}\left[e^{2p(Y_\cdot^+)_*}+e^{q(Y_\cdot^-)_*}+\int_0^T|Z_s|^2ds\right]<+\infty,
\end{align}
where $(Y_\cdot)_*:=\sup_{t\in[0,T]}|Y_t|$. Moreover, both 
\begin{align}\label{exinD1}
\left\{\exp\left(2pY_t^++2p(\gamma-1)\int_0^tr_s^-ds\right)\right\}_{t\in[0,T]}
\end{align}
and
\begin{align}\label{exinD2}
\left\{\exp\left(qY_t^-+q(\gamma-1)\int_0^t\left(r^+_s+\frac{1}{\gamma}\mu'_s\Sigma_s^{-1}\mu_s\right)ds\right)\right\}_{t\in[0,T]}
\end{align} 
are of class $(\rm D)$ with $p$ and $q$  given in Assumption \ref{11}. Furthermore, assume Assumption \ref{2} also holds, then $Y$ is bounded from above by a constant.
\end{proposition}
\begin{remark}
Motivated by \cite{fan2020uniqueness,fan2023user}, we use the test function method and the priori estimate technique, which are standard procedures for the existence of solutions to BSDEs. We find that BSDE (\ref{mar}) is concave with respect to $Y$; it is therefore natural to use the convex-dual method proposed by \cite{delbaen2011uniqueness} to prove the uniqueness of $(Y, Z)$. 
\end{remark}
\begin{remark}
It follows from (\ref{exin}) that the solution $Y$ to BSDE (\ref{mar}) is exponentially integrable, and $Z$ is square integrable, both under the original probability measure $\mathbb{P}$. We consider the positive and negative parts of $Y$ separately to examine its exponential integrability. In particular, the exponential integrability of $Y^+$ plays a crucial role in proving the martingale property in (\ref{marineq}). Moreover, $e^{Y^+}$ and $e^{Y^-}$ possess a stronger property of being uniformly integrable, as shown in (\ref{exinD1}) and (\ref{exinD2}).
\end{remark}

The following proposition presents the martingale optimal principle under  Assumption \ref{1}.  We reformulate it from the proof of Lemma B.1 and Theorem 2.14 of \cite{xing2017consumption} in our admissible class $\mathcal{A}$. It is sufficient to show the existence of the optimality that $G^{c,\pi}$ is a local supermartingale and $G^{c^*,\pi^*}$ is a local martingale, thanks to the special structure of $G^{c,\pi}$.

\begin{proposition}\label{g}
Suppose Assumption \ref{1} holds and  $\gamma, \psi>1$.  Let $(Y,Z)$ be the solution to BSDE (\ref{mar}). Then
\begin{enumerate}[{(i)}]
\item For any $(c,\pi)\in\mathcal{A}$, we have $G^{c,\pi}$ is a local supermartingale.

\item Define the candidate optimal strategy
\begin{align}\label{os}
c^*=\delta^\psi e^{-\frac{\psi}{\theta}Y}\mathcal{W}^*,\quad \pi^*=\frac{1}{\gamma}\Sigma^{-1}\left(\mu+\sigma Z'\right),
\end{align}
then $G^{c^*,\pi^*}$ is a local martingale. Furthermore, if $(c^*,\pi^*)\in\mathcal{A}$, then $(c^*,\pi^*)$ is the optimal strategy for problem (\ref{max}). 
\end{enumerate}
\end{proposition}
\begin{proof}
For any $(c,\pi)\in\mathcal{A}$,
recalling from (\ref{upo}) that
\begin{align}\label{up}
dG_t^{c,\pi}=&\frac{\mathcal{W}_t^{1-\gamma}}{1-\gamma}e^{Y_t}\Bigg[\left(Z_t+(1-\gamma)\pi_t'\sigma_t\right)dW_t+\big(\delta\theta e^{-\frac{Y_t}{\theta}}\hat{c}_t^{1-\frac{1}{\psi}}-(1-\gamma)\hat{c}_t+(1-\gamma)\pi_t'(\mu_t+\sigma_t Z_t')\notag\\
&-\frac{\gamma(1-\gamma)}{2})\pi_t'\Sigma_t\pi_t-H(t,Y_t,Z_t)
+(1-\gamma)r_t-\delta\theta+\frac{1}{2}Z_tZ_t'\big)dt\Bigg],
\end{align}
where $\hat{c}=\frac{c}{\mathcal{W}}$. The drift term in (\ref{up}) is a concave function of $(\hat c, \pi)$ and the maximum is achieved at $(c^*,\pi^*)$ in (\ref{os}) with the maximum value 0 thanks to the definition of  $H$ in (\ref{generator-complete}). Therefore, $G_\cdot^{c,\pi}$ is a local supermartingale for any $(c,\pi)\in\mathcal{A}$ and  a local martingale for  $(c^*,\pi^*)$. Taking an appropriate stopping time sequence $\{\varsigma_n, n\geq 1\}$ with $\varsigma_{n}\stackrel{a.s.}\longrightarrow T$, for any $(c,\pi)\in\mathcal{A}$, we have
\begin{align*}
\mathbf{E}\left[G_{\varsigma_n}^{c,\pi}\right]&=\mathbf{E}\left[\int_0^{\varsigma_n}f\left(c_s,\frac{\left(\mathcal{W}_s\right)^{1-\gamma}}{1-\gamma}e^{Y_s}\right)ds+\frac{\left(\mathcal{W}_{\varsigma_n}\right)^{1-\gamma}}{1-\gamma}e^{Y_{\varsigma_n}}\right]\\
& \leq G_0^{c,\pi}=G_0^{c^*,\pi^*}
=\mathbf{E}\left[G_{\varsigma_n}^{c^*,\pi^*}\right]=\mathbf{E}\left[\int_0^{\varsigma_n}f\left(c^*_s,\frac{\left(\mathcal{W}^*_s\right)^{1-\gamma}}{1-\gamma}e^{Y_s}\right)ds+\frac{\left(\mathcal{W}^*_{\varsigma_n}\right)^{1-\gamma}}{1-\gamma}e^{Y_{\varsigma_n}}\right].
\end{align*}
This implies that
\begin{align*}
&\mathbf{E}\left[\frac{(\mathcal{W}_{\varsigma_n})^{1-\gamma}}{1-\gamma}e^{Y_{\varsigma_n}}+\int_0^{\varsigma_n}\delta\frac{\left(c_s\right)^{1-\frac{1}{\psi}}}{1-\frac{1}{\psi}}\left((\mathcal{W}_s)^{1-\gamma}e^{Y_s}\right)^{1-\frac{1}{\theta}}ds\right]+\delta\theta\mathbf{E}\left[\int_0^{\varsigma_n}\frac{(\mathcal{W}_s^*)^{1-\gamma}}{1-\gamma}e^{Y_s}ds\right]\\
\leq &\mathbf{E}\left[\frac{(\mathcal{W}_{\varsigma_n}^*)^{1-\gamma}}{1-\gamma}e^{Y_{\varsigma_n}}+\int_0^{\varsigma_n}\delta\frac{\left(c_s^*\right)^{1-\frac{1}{\psi}}}{1-\frac{1}{\psi}}\left((\mathcal{W}_s^*)^{1-\gamma}e^{Y_s}\right)^{1-\frac{1}{\theta}}ds\right]+\delta\theta\mathbf{E}\left[\int_0^{\varsigma_n}\frac{(\mathcal{W}_s)^{1-\gamma}}{1-\gamma}e^{Y_s}ds\right].
\end{align*}
If $(c^*,\pi^*)\in\mathcal{A}$, sending $n\to\infty$, the monotone convergence theorem and the class (D) property of $(\mathcal{W})^{1-\gamma}e^{Y}$ and $(\mathcal{W}^*)^{1-\gamma}e^{Y}$ yields
\begin{align}\label{maxuti}
&\sup_{(c,\pi)\in\mathcal{A}}\mathbf{E}\left[\int_0^{T}f\left(c_s,\frac{\left(\mathcal{W}_s\right)^{1-\gamma}}{1-\gamma}e^{Y_s}\right)ds+\frac{\left(\mathcal{W}_{T}\right)^{1-\gamma}}{1-\gamma}\right]
\notag\\= &
\mathbf{E}\left[\int_0^{T}f\left(c^*_s,\frac{\left(\mathcal{W}^*_s\right)^{1-\gamma}}{1-\gamma}e^{Y_s}\right)ds+\frac{\left(\mathcal{W}^*_{T}\right)^{1-\gamma}}{1-\gamma}\right].
\end{align}
Then $(c^*,\pi^*)$ is the optimal strategy for problem (\ref{max}).
\end{proof}

By Proposition \ref{g},  
it remains to verify the candidate optimal strategy $(c^*,\pi^*)\in \mathcal{A}$, defined by \eqref{os}.  Different from the Markovian environments and Lyapunov arguments in \cite{xing2017consumption},  we verify it by the crucial martingale property along with BSDE techniques.  Thereby, our result admits the non-Markovian situation. The maximization (\ref{max}) is determined by Proposition \ref{g} once we complete the verification of the optimal strategy $(c^*,\pi^*)\in\mathcal{A}$. To achieve it, we need Assumption \ref{2} to obtain the admissibility of $(c^*,\pi^*)$.

\begin{theorem}\label{main}
Suppose Assumption \ref{1} holds and  $\gamma, \psi>1$.  Let $(Y, Z)$ be the solution to BSDE (\ref{mar}). Then the stochastic exponential 
\begin{align}\label{mse}
Q:=\mathscr{E}\left(\int\Big((1-\gamma)(\pi_s^*)'\sigma_s+Z_s\Big)dW_s\right)
\end{align}
is a true martingale, or equivalently, of class $(\rm D)$ under $\mathbb{P}$ where $\pi^*$ is defined in (\ref{os}). Furthermore, if Assumption \ref{2} holds as well, then $(c^*,\pi^*)\in\mathcal{A}$. For the initial wealth $\omega$, the optimal value function is given by
\begin{align}\label{main2}
V_0^{c^*,\pi^*}=\frac{\omega^{1-\gamma}}{1-\gamma}e^{Y_0}.
\end{align}
\end{theorem}

\begin{remark}
Note that the martingale property of $Q$ does not require $r$ to be bounded from below. Assumption \ref{2} ($r^-$ is bounded) is only needed in Step 2 of the proof of Theorem \ref{main} to show that the candidate strategy $(c^*, \pi^*)\in \mathcal{A}$. 

We work in a non-Markovian market, which is more sophisticated than the Markovian one.  Theorem \ref{main} guarantees the martingale property without the help of Lyapunov arguments utilized by \cite{xing2017consumption}.  The martingale property holds even without the condition on $r+\frac{\mu'\Sigma^{-1}\mu}{2\gamma}$, which ensures the upper boundedness of $Y$ and is crucial for Lyapunov's arguments. 

Compared with Lemma B.2 in \cite{xing2017consumption},  our proof is more straightforward, which only depends on the exponential integrability of $Y$ and the relative entropy method.  \cite{hu2018exponential} solved an exponential utility maximization problem with unbounded payoffs using a similar technique.
\end{remark}

\begin{remark}
Proposition \ref{g} holds for a large admissible class $\mathcal{A}$, which allows us to remove several technical conditions on market parameters such as Assumption 2.11 in \cite{xing2017consumption}. 
\end{remark}


\subsection{Convex dual for Epstein-Zin optimization}\label{convex-dual}
A dual problem for a consumption-investment optimization with Epstein-Zin utility is a minimization problem of a convex functional of state price densities. Recall the state price density process $D$ satisfies that $D_0=1$, $D>0$, $D\mathcal{W}+\int_0^\cdot D_sc_sds$ is a supermartingale for $(c,\pi)\in\widehat{\mathcal{A}}$ under $\mathbb{P}$, where 
$$\widehat{\mathcal{A}}=\left\{(c,\pi): V^{c,\pi} {\rm\ exists\ }, V^{c,\pi}<0, V^{c,\pi} {\rm\ is\ of\ class\ (D)}\right\}.$$
The set of all state price density processes is denoted by $\mathcal{D}$.  Compared with (\ref{adset}), $\widehat{\mathcal{A}}$ is larger due to the duality, which consists of all nonnegative self-financing wealth processes. Note that
$$D^0_t:=e^{-\int_0^tr_sds}\mathcal{E}\left(\int-\mu_s'\Sigma_s^{-1}\sigma_s dW_s\right)_t,~~~0\leq t\leq T,$$
is a state price density, which leads to no arbitrage in the market setting.

The dual version for the Epstein-Zin utility is extracted from Section 2.3 of \cite{matoussi2018convex}. For any $D\in \mathcal{D}$ and $y>0$, the Epstein-Zin stochastic differential dual for $yD$ over a time horizon $T$ is a semimartingale $U^{yD}$ that satisfies
\begin{align}\label{eq:dual-u}
U^{yD}_t=\mathbf{E}\left[\int_t^T g\left(yD_s, \frac{1}{\gamma}U_s^{yD}\right)ds+U(yD_T)\Big|\mathscr{F}_t\right],~~~\forall t\in[0,T], 
\end{align}
where 
\begin{align*}
&g(d,u)=\delta^\psi\frac{d^{1-\psi}}{\psi-1}((1-\gamma)u)^{1-\frac{\gamma\psi}{\theta}}-\delta\theta u,\quad d>0, u<0,\\
&U(d)=\frac{\gamma}{1-\gamma}d^{\frac{\gamma-1}{\gamma}},\quad d>0.
\end{align*}
\cite{matoussi2018convex} formulate $U^{yD}$ by the variational representation. Then they define the \emph{admissible} class of $D$ similar to $\widehat{\mathcal{A}}$:
$$\widehat{\mathcal{D}}=\left\{D\in\mathcal{D}:  U^{yD} {\rm\ exists\ }, U^{yD}<0, U^{yD} {\rm\ is\ of\ class\ D}\right\}.$$
There are sufficient conditions ensuring the existence of Epstein-Zin stochastic differential duals, which implies $\widehat{\mathcal{D}}\neq\emptyset$, see Proposition 2.5 of \cite{matoussi2018convex}.

Theorem 2.7 in \cite{matoussi2018convex} tells us that the Epstein-Zin stochastic differential dual is exactly the dual version for Epstein-Zin utility when $\gamma,\psi>1$:
\begin{align}\label{duality}
\sup_{(c,\pi)\in\widehat{\mathcal{A}}}V_0^{c,\pi}\leq\inf_{y>0}\left(\inf_{D\in\widehat{\mathcal{D}}}U_0^{yD}+wy\right).
\end{align}
We call the right side of (\ref{duality}) the dual problem. In fact, (\ref{duality}) holds for a broader range of $\gamma$ and $\psi$: $\gamma\psi\geq1,\psi>1$ or $\gamma\psi\leq1,\psi<1$. Since both \cite{xing2017consumption} and \cite{matoussi2018convex} are referenced in our analysis, we have adopted the common parameter range addressed in both works, namely, $\gamma,\psi>1$.  

\cite{matoussi2018convex} show that the inequality \eqref{duality} is
actually an equality, i.e., there is no duality gap,  under the non-Markovian market with bounded market parameters (Theorem 3.6) or the Markovian market (Theorem 3.12). In the case of bounded market coefficients, Theorem 5.1 of \cite{kraft2017optimal} also implies the absence of a duality gap.

In this subsection,  we demonstrate that there is no duality gap in the non-Markovian market with a class of unbounded market parameters. Since the primal problem has been solved in subsection \ref{mop}, we focus on obtaining the results on the dual side.

\begin{theorem}\label{dual}
Suppose that Assumptions \ref{1} and \ref{2} hold and $\gamma,\psi>1$. Then there exist $y^*>0$ and $D^*\in\widehat{\mathcal{D}}$ such that
\begin{align*}
\max_{(c,\pi)\in\widehat{\mathcal{A}}}V_0^{c,\pi}=\max_{(c,\pi)\in\mathcal{A}}V_0^{c,\pi}=V_0^{c^*,\pi^*}=U_0^{y^*D^*}+\omega y^*=\min_{y>0}\left(\min_{D\in\widehat{\mathcal{D}}}\left(U_0^{yD}+wy\right)\right),
\end{align*}where $(c^*,\pi^*)$ is defined by \eqref{os}.
\end{theorem}

\begin{remark}
The first equality implies that the admissible class can be generalized from $\mathcal{A}$ to $\widehat{\mathcal{A}}$ through the duality theorem \ref{dual}, wherein the integrability requirements are further relaxed.

$R^{yD}$, defined as (\ref{upo2}), is established under the martingale optimal principle for the dual problem similar to $G^{c,\pi}$. Naturally, the dual problem shares the same stochastic exponential $Q$ with the primal problem. Hence the proof can be followed by a similar method to that used in the primal problem. 
\end{remark}

\section{Examples}\label{Example}
In this section, we discuss three well-known models (Heston, linear diffusion, and CIR)  in a non-Markovian environment, and show that the parameters of these models satisfy the exponential integrability condition, so  Assumptions \ref{1} and \ref{2} hold. We assume $\widetilde{W}$ is a standard Brownian motion, independent of $W$, and $\rho\in[-1,1]$ is the correlation coefficient. 

\subsection{Heston model}
Consider a stochastic volatility
model in which the stock price $S$ and a 1-dimensional state variable $X$ follow
\[
\begin{dcases}
dX_t=b(l-X_t)dt+a\sqrt{X_t}dW_t,\\
dS_t=S_t\left[\left(r_t+\mu(X_t)\right)dt+\sqrt{X_t}\sigma\rho dW_t+\sqrt{X_t}\sigma\sqrt{1-\rho^2}d\widetilde{W}_t\right],
\end{dcases}
\]
where $r$ is a bounded adapted process,  $\mu(x)=\lambda x$, and  $\lambda, a, b,l,\sigma$ are positive constants. The inverse Heston model has been studied by \cite{chacko2005dynamic} for recursive utilities with $\psi=1$. The following result provides parameter conditions such that Assumptions \ref{1} and \ref{2} hold.

\begin{proposition}\label{heston}
Suppose that $\gamma,\psi>1$ and $q>1$ is a constant. Moreover, 
\begin{enumerate}[{(i)}]
  \item Either $4bl\leq a^2, qT\lambda^2a^2(e^{bT}-1)<2b\sigma^2$ 
  \item or $2bl>a^2, 2q\lambda^2a^2<b^2\sigma^2$. 
\end{enumerate}
Then Assumption \ref{3} holds.
\end{proposition}
\begin{proof}
Since $r$ is a bounded adapted process, it is clear that
\begin{align*}
\mathbf{E}\left[\exp\left(q(\gamma-1)\int_0^Tr_s^+ ds\right)\right]<+\infty.
\end{align*}
It then follows from Theorem 4.1 of \cite{yong2004some} that Proposition \ref{heston}(i) makes
\begin{align*}
\mathbf{E}\left[\exp\left(q\int_0^T\mu'_s\Sigma_s^{-1}\mu_sds\right)\right]=\mathbf{E}\left[\exp\left(q\frac{\lambda^2}{\sigma^2}\int_0^TX_sds\right)\right]<+\infty,
\end{align*}
which can also be guaranteed by Proposition \ref{heston}(ii) due to Lemma C.1 of \cite{xing2017consumption}.
\end{proof}
 
In Proposition 3.2 of \cite{xing2017consumption}, all parameter assumptions are imposed to satisfy the necessary conditions for verifying the main results. Among these assumptions, one of them is made to fulfill the requirements of an auxiliary lemma (Lemma C.1) introduced in their Appendix. Besides this lemma, we find Theorem 4.1 of \cite{yong2004some} to verify the exponential integrability of market parameters. Therefore, we have two restrictions on the parameters and either one suffices. 

Note that $r$ makes the model non-Markovian. Moreover, if the time terminal $T$ is sufficiently small, the second condition of item (i) holds automatically and is not needed.

We thank the referees for their suggestions, which motivate us to add Remark 4.1 validating two popular Markovian models, and Remark 4.2 analyzing two non-Markovian models discussed in the financial literature.

\begin{remark} Assumption \ref{3} also holds for some other popular volatility models, including \cite{stein1991stock} and \cite{chacko2005dynamic}. 

For the model of \cite{stein1991stock},
\[
\begin{dcases}
d\sigma_t=-\delta(\sigma_t-\theta)dt+kdW_t,\\
dS_t=S_t(\mu dt+\sigma_td\widetilde{W}_t),
\end{dcases}
\]
where $r,\delta,\theta,k,\mu$ are fixed constants.  What we need to verify is that $$\mathbf{E}\left[\exp\left(q\int_{0}^{T}\frac{(\mu-r)^2}{\sigma_t^2}dt\right)\right]<\infty.$$
This can be established by using Theorem 3.1(ii) of \cite{yong2004some}, which requires no additional conditions since the power $\gamma$ of the function $\varphi$ in that theorem can be set to 1. In other words, Assumption \ref{3} is satisfied automatically in the model of \cite{stein1991stock}.

For the model of \cite{chacko2005dynamic},
\[
\begin{dcases}
dy_t=\kappa(\theta-y_t)dt+\sigma\sqrt{y_t}dW_t,\\
dS_t=S_t\left(\mu dt+\sqrt{\frac{1}{y_t}}\rho dW_t+\sqrt{\frac{1}{y_t}}\sqrt{1-\rho^2}d\widetilde{W}_t\right),
\end{dcases}
\]
where $r,\kappa,\theta,\sigma,\mu$ are fixed constants. Our goal is to verify that
$$\mathbf{E}\left[\exp\left(q\int_{0}^{T}(\mu-r)^2y_t^2dt\right)\right]<\infty.$$
This follows from Lemma C.1 of \cite{xing2017consumption}, which requires the condition 
$q(\mu-r)^2<\frac{\kappa^2}{2\sigma^2}$
to make Assumption \ref{3} hold.
\end{remark}

\begin{remark}
There are some examples of non-Markovian models relevant for the finance literature.
 The first example is the \textit{rough fractional stochastic volatility   model} proposed by \cite{gatheral2018volatility}:
\begin{align*}
&\frac{dS_t}{S_t}=\mu_tdt+\sigma_tdW_t,\\
&\sigma_t=e^{X_t},
\end{align*} 
where $X$ is a fractional OU process satisfying for $\alpha,\nu>0$,
$$dX_t=-\alpha X_tdt+\nu dW_t^H,$$
and  $W^H$ is a fractional Brownian motion with Hurst parameter $H\in(0,1)$ and has  the Mandelbrot van Ness representation
\begin{align*}
W_t^H=&\frac{1}{\Gamma(H+1/2)}\int_{-\infty}^{0}\left((t-s)^{H-\frac{1}{2}}-(-s)^{H-\frac{1}{2}}\right)dW_s+\frac{1}{\Gamma(H+1/2)}\int_{0}^{t}(t-s)^{H-\frac{1}{2}}dW_s.
\end{align*}
This model is stationary due to the log-volatility being a fractional OU process. The introduction of $W^H$ makes the model non-Markovian. Moreover, the solution $X_t$ is given by
$$X_t=\nu\int_{-\infty}^{t}e^{-\alpha(t-s)}dW_t^H.$$
Assumption \ref{3}  requires that $r^-$ is bounded and for some $q>1$,
\begin{align*}
\mathbf{E}\left[\exp\left(q(\gamma-1)\int_0^Tr_s^+ ds\right)+\exp\left(q\int_0^T\frac{\mu_s^2}{\sigma_s^2}ds\right)\right]<+\infty.
\end{align*}
However, since $\sigma$ itself is an exponential, we need to verify the exponential integrability of an exponential $\exp\left(q\int_0^T\mu_s^2e^{-2X_s}ds\right)$.

The second example is the \textit{rough Heston model} proposed by \cite{el2019characteristic}:
\begin{align*}
&\frac{dS_t}{S_t}=\sqrt{V_t}dW_t,\\
&V_t=V_0+\frac{1}{\Gamma(\alpha)}\int_{0}^{t}(t-s)^{\alpha-1}\gamma(\theta-V_s)ds
+\frac{1}{\Gamma(\alpha)}\int_{0}^{t}(t-s)^{\alpha-1}\gamma\nu\sqrt{V_s}dB_s,
\end{align*} 
where $\gamma,\theta,\nu>0$, $\alpha\in(1/2,1)$ (when $\alpha=1$, it is the classical Heston model), and $W$ and $B$ are two correlated Brownian motions. Inspired by the formulation of $W^H$,  this model introduces the kernel $(t-s)^{\alpha-1}$ in a Heston-like stochastic volatility process as above. To apply our results, we need to verify the exponential integrability of $V$.  We leave the above two non-Markovian models for future research. 
\end{remark}

\subsection{Linear diffusion model}
The following model has a 1-dimensional state variable following an Ornstein-Uhlenbeck process, which constitutes a linear excess return of risky assets
\[
\begin{dcases}
dX_t=-bX_tdt+adW_t,\\
dS_t=S_t\left[\left(r_t+\mu(X_t)\right)dt+\sigma\rho dW_t+\sigma\sqrt{1-\rho^2}d\widetilde{W}_t\right],
\end{dcases}
\]
where $r$ is a bounded adapted process,  $\mu(x)=\sigma(\lambda_0+\lambda_1x)$, and $a, b,\sigma,\lambda_0,\lambda_1$ are positive constants. If the volatility term is a stochastic process while the drift is a constant, the model has been studied by \cite{stein1991stock}. Then the following proposition identifies market parameters such that  Assumption \ref{1} and \ref{2} hold.
\begin{proposition}\label{LD}
Suppose that $\gamma,\psi>1$ and $q>1$ is a constant. Moreover, 
\begin{enumerate}[{(i)}]
  \item Either for some $c>2q\lambda_1^2$,
      \begin{align}\label{LDcond}
         \frac{cTa^2}{b}(e^{2bT}-1)<1
      \end{align} 
  \item or $2q\lambda_1^2<\frac{b^2}{2a^2}$.
\end{enumerate}
Then Assumption \ref{3} holds.
\end{proposition}
\begin{proof}
Since $r$ is bounded, we only have to verify that
$$\mathbf{E}\left[\exp\left(q\int_0^T\mu'_s\Sigma_s^{-1}\mu_sds\right)\right]<+\infty.$$
A simple calculation indicates that
\begin{align*}
&\mathbf{E}\left[\exp\left(q\int_0^T\mu'_s\Sigma_s^{-1}\mu_sds\right)\right]\leq\exp\left(2q\lambda_0^2T \right)\mathbf{E}\left[\exp\left(2q\lambda_1^2\int_{0}^{T}X^2_sds\right)\right].
\end{align*}
According to Theorem 3.1(ii) of \cite{yong2004some} Assumption \ref{1} is verified by (\ref{LDcond}).
\end{proof}

Proposition 3.4 of \cite{xing2017consumption} also considered the linear diffusion model and added some assumptions similar to Proposition 3.2 of \cite{xing2017consumption}. We can set $a,b$ as two adapted processes satisfying (\ref{LDcond}) and Proposition \ref{LD} still holds, then our linear diffusion model permits a broader range of parameter settings.

\subsection{CIR model}
In this model, the short interest rate $r$ satisfies the CIR model. Specifically, the model dynamics are as follows
\[
\begin{dcases}
dr_t=(b-lr_t)dt+ar_t^{\frac{1}{2}} dW_t,\\
dS_t=S_t\left[\left(r_t+\mu_t\right)dt+\sigma_t\rho dW_t+\sigma_t\sqrt{1-\rho^2}d\widetilde{W}_t\right],
\end{dcases}
\]
where $\mu$ is a bounded adapted process and $\varepsilon\leq\sigma_t\leq\frac{1}{\varepsilon}$ for all $t\geq0$ a.s. for constant $0<\varepsilon<1$, and  $a, b,l$ are positive constants. 

The main difference of this model from the previous two models lies in the fact that the market price of risk $\left(\mu'\Sigma^{-1}\mu\right)^{\frac{1}{2}}$ is bounded, while the interest rate is stochastic without boundedness restrictions. Then the following result specifies Assumptions \ref{1} and \ref{2} to explicit model parameter conditions. 


\begin{proposition}\label{CIR}
Suppose that $\gamma,\psi>1$ and $q>1$ is a constant. Moreover, 
\begin{enumerate}[{(i)}]
  \item $2b>a^2$;
  \item $q(\gamma-1)<\frac{l^2}{2a^2}$.
\end{enumerate}
Then Assumption \ref{3} holds.
\end{proposition}
\begin{proof}
Condition (i) implies that $r$ is positive. Since $\mu'\Sigma^{-1}\mu$ is bounded, the only assumption we need to verify is 
$$\mathbf{E}\left[\exp\left(q(\gamma-1)\int_0^Tr_s^+ ds\right)\right]<+\infty,$$ 
which is ensured by Lemma C.1 of \cite{xing2017consumption} with Proposition \ref{CIR}(i)(ii).   
\end{proof}

CIR model is a classic non-Markovian stochastic interest model. Its well-posedness is studied by a large body of literature; see \cite{cox1985theory}. As we assume that the excess rate of return and the volatility are both bounded processes, the unbounded parameter is the interest rate. The role of conditions (i) and (ii) is to ensure $r$ is positive and exponentially integrable.

\section{The proofs}
In this section, we provide the proofs of Proposition \ref{Y}, Theorem \ref{main}, and Theorem \ref{dual}.

\subsection{The proof of Proposition \ref{Y}}
We first prove the existence result of the solution to BSDE (\ref{mar}). For any integers $m,n,k\geq 1$,  for each $(t,y,z)\in[0,T]\times\mathbb{R}\times\mathbb{R}^{1\times n}$, we truncate the generator as follows:
\begin{align}\label{generatortruncated}
H^m(t,y,z):=&z\left(\frac{1}{2}I_n+\frac{1-\gamma}{2\gamma}\sigma_t'\Sigma_t^{-1}\sigma_t\right)z'+\frac{1-\gamma}{\gamma}\mu_t'\Sigma_t^{-1}\sigma_tz'+\frac{\theta}{\psi}\delta^\psi \left(e^{-\frac{\psi}{\theta}y}\wedge m\right)\notag\\
&+\frac{1-\gamma}{2\gamma}\mu_t'\Sigma_t^{-1}\mu_t+(1-\gamma) r_{t}-\delta \theta,\\
H^{m,n,k}(t,y,z):=&\left(H^m(t,y,z)\right)^{n,k}:=\mathbb{I}_{\{t\leq\rho_n\}}(H^m)^+(t,y,z)-\mathbb{I}_{\{t\leq\rho_k\}}(H^m)^-(t,y,z),\notag
\end{align}
where the stopping times
\begin{align*}
\rho_i:=\inf\left\{t\geq0: \int_0^t|r_s|ds+\int_0^t\mu_s'\Sigma_s^{-1}\mu_sds\geq i\right\}\wedge T,\quad i=n,k.
\end{align*} 
Note that $\lim_{m,n,k\to\infty}H^{m,n,k}=H$ and $H^{m,n,k}$ is Lipschitz in $y$ and quadratic in $z$. In fact, since the eigenvalues of $\sigma'\Sigma^{-1}\sigma$ are either $0$ or $1$, it implies that 
\begin{align}\label{zi}
\frac{1}{2\gamma}|z|^2\leq\frac{1}{2}|z|^2+\frac{1-\gamma}{2\gamma}z\sigma'\Sigma^{-1}\sigma z'\leq\frac{1}{2}|z|^2.
\end{align}
Then, we can derive that
\begin{align*}
&z\left(\frac{1}{2}I_n+\frac{1-\gamma}{2\gamma}\sigma_t'\Sigma_t^{-1}\sigma_t\right)z'+\frac{1-\gamma}{\gamma}\mu_t'\Sigma_t^{-1}\sigma_tz'\\
\leq&\frac{1}{2}|z|^2+\frac{\gamma-1}{\gamma}\left|\mu_t'\Sigma_t^{-1}\sigma_t\right||z|\leq\frac{2\gamma-1}{2\gamma}|z|^2+\frac{\gamma-1}{2\gamma}\mu_t'\Sigma_t^{-1}\mu_t
\end{align*}
and
\begin{align*}
&z\left(\frac{1}{2}I_n+\frac{1-\gamma}{2\gamma}\sigma_t'\Sigma_t^{-1}\sigma_t\right)z'+\frac{1-\gamma}{\gamma}\mu_t'\Sigma_t^{-1}\sigma_tz'\\
\geq&\frac{1}{2\gamma}|z|^2-\frac{\gamma-1}{\gamma}\left|\mu_t'\Sigma_t^{-1}\sigma_t\right||z|\geq\frac{2-\gamma}{2\gamma}|z|^2-\frac{\gamma-1}{2\gamma}\mu_t'\Sigma_t^{-1}\mu_t.
\end{align*}
It follows from Theorem 2.3 of \cite{kobylanski2000backward} that BSDE whose terminal condition is $0$ and generator is $H^{m,n,k}$ admits\footnote{ ${S}^\infty$ denotes the space of $1$-dimensional continuous adapted processes $Y$ such that $||\sup_{0\leq s\leq T}|Y_s|||_{\infty}<+\infty$,  and $\mathcal{M}^2$ denotes the space of predictable processes $Z$ such that $\mathbf{E}\left[\int_0^T|Z_s|^2ds\right]<+\infty$.} a solution  $(Y^{m,n,k}, Z^{m,n,k})\in\mathcal{S}^\infty\times\mathcal{M}^2$. 

Motivated by \cite{fan2023user}, we construct a test function $\phi$.  More sepcifically, for each $(s,x)\in[0,T]\times[0,+\infty)$, $h:[0,T]\to[0,+\infty)$ and $\lambda>0$, define the function
\begin{align}\label{eq:test}
    \phi(s,x; h,\lambda):=\exp\left(\lambda x+\lambda\int_0^sh(r)dr\right),
\end{align}
which satisfies 
\begin{align}\label{psi0}
\left\{
\begin{aligned}
&\phi_s(s,x;h,\lambda)-h(s)\phi_x(s,x;h,\lambda)=0,\\
&\phi_{xx}(s,x;h,\lambda)-\lambda\phi_x(s,x;h,\lambda)=0,\\
&\phi(s,x;h,\lambda)\geq \lambda x\Rightarrow\phi_x(s,x;h,\lambda)\geq \lambda.
\end{aligned}
\right.
\end{align} 
Since for each $m,n,k$ and any $(s,y,z)$,
\begin{align*}
&\mathbb{I}_{\{y>0\}}H^{m,n,k}(s,y,z)\leq\frac{2\gamma-1}{2\gamma}\mathbb{I}_{\{Y_s>0\}}|z|^2+(\gamma-1)r_s^--\delta\theta,\\
&\mathbb{I}_{\{y\leq0\}}H^{m,n,k}(s,y,z)
\geq\frac{1-\gamma}{2\gamma}\mathbb{I}_{\{Y_s\leq0\}}|z|^2+(1-\gamma)r^+_s+\frac{1-\gamma}{\gamma}\mu'_s\Sigma_s^{-1}\mu_s+\theta\frac{\delta^\psi}{\psi}.
\end{align*}
Applying It${\rm\hat{o}}$-Tanaka's formula to $\phi$ and using the fact (\ref{psi0}), we have
\begin{align*}
&d\phi(s,(Y_s^{m,n,k})^+;(\gamma-1)r_\cdot^--\delta\theta,2p)\\
\geq & p\left(2p-\frac{2\gamma-1}{\gamma}\right)\mathbb{I}_{\{Y_s^{m,n,k}>0\}}|Z^{m,n,k}_s|^2ds
+\phi_x(s,(Y_s^{m,n,k})^+;(\gamma-1)r_\cdot^--\delta\theta,2p) \mathbb{I}_{\{Y^{m,n,k}_s>0\}}Z^{m,n,k}_sdW_s,
\end{align*}and
\begin{align*}
&d\phi\left(s,(Y^{m,n,k}_s)^-;-\left[(1-\gamma)r^+_\cdot+\frac{1-\gamma}{\gamma}\mu'_\cdot\Sigma_\cdot^{-1}\mu_\cdot+\theta\frac{\delta^\psi}{\psi}\right],q\right)\\
\geq &\frac{1}{2}q\left(q-\frac{\gamma-1}{\gamma}\right)\mathbb{I}_{\{Y^{m,n,k}_s\leq0\}}|Z^{m,n,k}_s|^2ds\\
& -\phi_x\left(s,(Y^{m,n,k}_s)^-;-\left[(1-\gamma)r^+_\cdot+\frac{1-\gamma}{\gamma}\mu'_\cdot\Sigma_\cdot^{-1}\mu_\cdot+\theta\frac{\delta^\psi}{\psi}\right],q\right)\mathbb{I}_{\{Y^{m,n,k}_s\leq0\}}Z^{m,n,k}_sdW_s.
\end{align*}
By Assumption \ref{11}, then it implies that $\phi(T,0;(\gamma-1)r_\cdot^--\delta\theta,2p)$ and $\phi(T,0;(\gamma-1)r^+_\cdot+\frac{\gamma-1}{\gamma}\mu'_\cdot\Sigma_\cdot^{-1}\mu_\cdot-\theta\frac{\delta^\psi}{\psi},q)$ are integrable. Thus, by Fatou's lemma and the convergence theorem, we can derive that  
\begin{align}
2p(Y^{m,n,k}_t)^+\leq&\phi(t,(Y^{m,n,k}_t)^+;(\gamma-1)r_\cdot^--\delta\theta,2p)\notag\\
&+p\left(2p-\frac{2\gamma-1}{\gamma}\right)\mathbf{E}\left[\int_t^T\mathbb{I}_{\{Y^{m,n,k}_s>0\}}|Z^{m,n,k}_s|^2ds\Big|\mathscr{F}_t\right]\notag\\
\leq&\mathbf{E}\left[\phi(T,0;(\gamma-1)r_\cdot^--\delta\theta,2p)\big|\mathscr{F}_t\right]=:\overline{Y}_t,\label{psi}\\
q(Y^{m,n,k}_t)^-\leq&\phi(t,(Y^{m,n,k}_t)^-;(\gamma-1)r^+_\cdot+\frac{\gamma-1}{\gamma}\mu'_\cdot\Sigma_\cdot^{-1}\mu_\cdot-\theta\frac{\delta^\psi}{\psi},q)\notag\\
&+\frac{1}{2}q\left(q-\frac{\gamma-1}{\gamma}\right)\mathbf{E}\left[\int_t^T\mathbb{I}_{\{Y^{m,n,k}_s\leq0\}}|Z^{m,n,k}_s|^2ds\Big|\mathscr{F}_t\right]\notag\\
\leq&\mathbf{E}\left[\phi(T,0;(\gamma-1)r^+_\cdot+\frac{\gamma-1}{\gamma}\mu'_\cdot\Sigma_\cdot^{-1}\mu_\cdot-\theta\frac{\delta^\psi}{\psi},q)\Big|\mathscr{F}_t\right]=:\underline{Y}_t.\label{psi2} 
\end{align}
Note that the choice of $p$ and $q$ in Assumption \ref{1} ensures that $p>\frac{2\gamma-1}{2\gamma}$ and $q>\frac{\gamma-1}{\gamma}$, hence the first inequalities of \eqref{psi} and \eqref{psi2} hold. The two inequalities indicate $\overline{Y}$ and $\underline{Y}$  are both class (D) independent of $m,n,k$. Since $H^{m+1,n,k+1}\leq H^{m+1,n,k}\leq H^{m,n,k}\leq H^{m,n+1,k}$, the comparison result Theorem 2.6 in \cite{kobylanski2000backward} implies that $Y^{m,n,k}$ is nonincreasing with $m$ and $k$ and nondecreasing with $n$. 

Furthermore, to apply the localization technique in Theorem 6 of \cite{briand2006bsde}, we introduce the following stopping time: for integer $i\geq1$, $\iota_i:=\inf\{t\in[0, T]: \max\{\overline{Y}_t,\underline{Y}_t\}\geq i\}\wedge T$ with the convention $\inf\emptyset=+\infty$. The assumption \ref{1} ensures that both $\overline{Y}_t$ and $\underline{Y}_t$ are integrable. Thus} applying the localization technique in Theorem 6 of \cite{briand2006bsde}, inequalities \eqref{psi} and \eqref{psi2}    allow us to construct a solution $(Y,Z)$ to BSDE (\ref{mar}) such that for each $t\in[0,T]$,
$$\inf_{m}\inf_{k}\sup_{n}(Y_t^{m,n,k})^+=:Y_t^+ \text{~~~and~~~}\inf_{m}\inf_{k}\sup_{n}(Y_t^{m,n,k})^-=:Y_t^-.$$
Meanwhile, 
$$\left\{\phi(t,Y_t^+;(\gamma-1)r_\cdot^--\delta\theta, 2p)\right\}_{t\in[0,T]}$$ and $$\left\{\phi\left(t,Y_t^-;(\gamma-1)r^+_\cdot+\frac{\gamma-1}{\gamma}\mu'_\cdot\Sigma_\cdot^{-1}\mu_\cdot-\theta\frac{\delta^\psi}{\psi},q\right)\right\}_{t\in[0,T]}$$ are of class (D), which implies both  $$\left\{\exp\left(2pY_t^++2p(\gamma-1)\int_0^tr_s^-ds\right)\right\}_{t\in[0,T]}$$ 
and 
$$\left\{\exp\left(qY_t^-+q(\gamma-1)\int_0^t\left(r^+_s+\frac{1}{\gamma}\mu'_s\Sigma_s^{-1}\mu_s\right)ds\right)\right\}_{t\in[0,T]}$$ 
are of class (D). 

Secondly we present the integrability of $(Y,Z)$. The above estimate results also hold for some sufficiently large parameters $\tilde{p}<p$ and $\tilde{q}<q$. Therefore, applying Doob's maximal inequality to 
$\phi(t,Y_t^+;(\gamma-1)r_\cdot^--\delta\theta,2\tilde{p})$, we have
$$\mathbf{E}\left[e^{2p(Y_\cdot^+)_*+2p\int_0^T((\gamma-1)r_s^--\delta\theta) ds}\right]\leq c_{\tilde{p},p} e^{-2p\delta\theta T}\mathbf{E}\left[e^{2p(\gamma-1)\int_0^Tr_s^-ds}\right]<+\infty,$$where $c_{\tilde{p},p}$ is a suitable constant. Similarly
$$\mathbf{E}\left[e^{q(Y_\cdot^-)_*+q\int_0^T((\gamma-1)r^+_s+\frac{\gamma-1}{\gamma}\mu'_s\Sigma_s^{-1}\mu_s-\theta\frac{\delta^\psi}{\psi} )ds}\right]<+\infty.$$
Then Assumption \ref{1} implies that $\mathbf{E}\left[e^{2p(Y_\cdot^+)_*}+e^{q(Y_\cdot^-)_*}\right]<+\infty$.
The square integrability of $Z$ comes from the integrability of $e^{2p(Y_\cdot^+)_*}$, $e^{q(Y_\cdot^-)_*}$ and \eqref{psi} and \eqref{psi2}.

Thirdly we derive the uniqueness result of the solution to BSDE (\ref{mar}), which is a consequence of the following Lemma. 

\begin{lemma}\label{unique} 
Suppose that $\gamma, \psi>1$ and Assumption \ref{1} holds. Then BSDE (\ref{mar}) admits a unique solution $(Y,Z)$ satisfying (\ref{exin}).
\end{lemma} 

\begin{proof} 

For each $(t,y)\in[0,T]\times \mathbb{R}$, since $H(t,y,\cdot)$ is a convex function,  we can define its Legendre-Fenchel transformation as follows.
\begin{align*}
J(t,y,l):=&\inf_{z\in \mathbb{R}^{1\times n}}(H(t,y,z)-zl)\\
=&-\frac{1}{2}\left(l'+\frac{\gamma-1}{\gamma}\mu_t'\Sigma_t^{-1}\sigma_t\right)\left(I_n+\frac{1-\gamma}{\gamma}\sigma'_t\Sigma_t^{-1}\sigma_t\right)^{-1}\left(l+\frac{\gamma-1}{\gamma}\sigma_t'\Sigma_t^{-1}\mu_t\right)\\
&+\frac{\theta}{\psi}\delta^\psi e^{-\frac{\psi}{\theta}y}+\frac{1-\gamma}{2\gamma}\mu_t'\Sigma_t^{-1}\mu_t+(1-\gamma) r_{t}-\delta \theta, ~~~~\forall l\in\mathbb{R}^n.
\end{align*}
Given a set
\begin{align*}
\mathscr{A}:=\Bigg\{ u ~\Bigg | ~
\begin{array}{ll}
\int_0^T|u_s|^2ds<+\infty, \mathbb{P}\text{-}a.s., ~
M:=\mathscr{E}\left(\int u_sdW_s\right)\text{is a martingale},\\\frac{d\mathbb{Q}}{d\mathbb{P}}:=M_T, ~~\mathbf{E}^{\mathbb{Q}}\left[\int_{0}^{T}|u_s|^2ds\right]<+\infty.
\end{array}\Bigg\}
\end{align*}
Let $u\in\mathscr{A}$ then $dW_t^u:=dW_t-u_tdt$ is a Brownian motion under $\mathbb{Q}$.  Since the eigenvalues of $\sigma'\Sigma^{-1}\sigma$ are either $0$ or $1$, then the eigenvalues of $\left(I_n+\frac{1-\gamma}{\gamma}\sigma'_t\Sigma_t^{-1}\sigma_t\right)^{-1}$ are $\gamma$ or $1$, which implies
\begin{align}\label{Jestimate}
&\left(u_t'+\frac{\gamma-1}{\gamma}\mu_t'\Sigma_t^{-1}\sigma_t\right)\left(I_n+\frac{1-\gamma}{\gamma}\sigma'_t\Sigma_t^{-1}\sigma_t\right)^{-1}\left(u_t+\frac{\gamma-1}{\gamma}\sigma_t'\Sigma_t^{-1}\mu_t\right)\notag\\
\geq&|u_t|^2+\frac{(\gamma-1)^2}{\gamma^2}\mu_t'\Sigma_t^{-1}\mu_t+\frac{2(\gamma-1)}{\gamma}\mu_t'\Sigma_t^{-1}\sigma_tu_t\notag\\
\geq&|u_t|^2+\frac{(\gamma-1)^2}{\gamma^2}\mu_t'\Sigma_t^{-1}\mu_t+\frac{2(1-\gamma)}{\gamma}|\mu_t'\Sigma_t^{-1}\sigma_t||u_t|\notag\\
\geq&|u_t|^2+\frac{(\gamma-1)^2}{\gamma^2}\mu_t'\Sigma_t^{-1}\mu_t+\frac{1-\gamma}{\gamma}\left(\mu_t'\Sigma_t^{-1}\mu_t+|u_t|^2\right)\notag\\
=&\frac{1}{\gamma}|u_t|^2+\frac{1-\gamma}{\gamma^2}\mu_t'\Sigma_t^{-1}\mu_t
\end{align}
and
\begin{align}\label{Jestimate2}
&\left(u_t'+\frac{\gamma-1}{\gamma}\mu_t'\Sigma_t^{-1}\sigma_t\right)\left(I_n+\frac{1-\gamma}{\gamma}\sigma'_t\Sigma_t^{-1}\sigma_t\right)^{-1}\left(u_t+\frac{\gamma-1}{\gamma}\sigma_t'\Sigma_t^{-1}\mu_t\right)\notag\\
\leq&\gamma\left(u_t'+\frac{\gamma-1}{\gamma}\mu_t'\Sigma_t^{-1}\sigma_t\right)\left(u_t+\frac{\gamma-1}{\gamma}\sigma_t'\Sigma_t^{-1}\mu_t\right)\notag\\
\leq&2\gamma|u_t|^2+2\frac{(\gamma-1)^2}{\gamma}\mu_t'\Sigma_t^{-1}\mu_t.
\end{align}
By Fenchel's inequality, we can get that  
\begin{align}\label{uniqueness}
\mathbf{E}^{\mathbb{Q}}\left[\int_{0}^{T}\mu_s'\Sigma_s^{-1}\mu_s+r_sds\right]<+\infty,
\end{align}
from Assumption \ref{1} and the fact that $u\in\mathscr{A}$.

Then under (\ref{Jestimate}), (\ref{Jestimate2}) and (\ref{uniqueness}), we can apply Proposition 6.4 in \cite{briand2003lp}: there exists a pair of processes $(Y^u,Z^u)$ to BSDE
$$Y_t^u=\int_t^TJ(s,Y_s^u,u_s)ds-\int_t^TZ_s^udW_s^u,\ t\in[0,T],$$
such that $\int_0^T|Z_s^u|^2ds<+\infty$ and  $\int_0^T|J(s,Y_s^u,u_s)|ds<+\infty$, $\mathbb{P}$-a.s. and $(Y_t^u)_{t\in[0,T]}$ belongs to the class (D) under $\mathbb{Q}$.

The uniqueness of $(Y, Z)$ follows if we can show that 
$$Y=\operatorname{essinf}_{u\in\mathscr{A}}Y^u.$$
To this end, for any $n\geq1$, introduce the stopping time: 
$$\tau_n:=\left\{s\geq t: \int_t^s|Z_r|^2dr+\int_t^s|Z_r^u|^2dr+\int_t^s|u_r|^2dr>n\right\}\wedge T.$$
Since $(Y,Z)$ satisfies  
\begin{align*}
Y_t=\int_t^TL(s,Y_s,Z_s)ds-\int_t^TZ_sdW_s^u,~~~~t\in[0,T],
\end{align*}
where $L(s,y,z):=H(s,y,z)-zu_s\geq J(s,y,u_s)$.
Then applying It${\rm\hat{o}}$'s formula to $(Y_s^u-Y_s)e^{\int_t^sh_rdr}$ with $$h_u:=\frac{L(s,Y_s^u,Z_s)-L(s,Y_s,Z_s)}{Y_s^u-Y_s}\mathbf{1}_{Y_s^u-Y_s\neq0}<0,$$ 
we get
\begin{align*}
(Y_{\tau_n}^u-Y_{\tau_n})e^{\int_t^{\tau_n}h_sds}-(Y_t^u-Y_t)=&\int_t^{\tau_n}e^{\int_t^sh_rdr}(J(s,Y_s^u,u_s)-L(s,Y_s^u,Z_s))ds\\
&+\int_t^{\tau_n}e^{\int_t^sh_rdr}(Z_s^u-Z_s)dW_s^u.
\end{align*}
Thus 
$$Y_t^u-Y_t\geq\mathbf{E}^{\mathbb{Q}}\left[(Y_{\tau_n}^u-Y_{\tau_n})e^{\int_t^{\tau_n}h_sds}|\mathscr{F}_t\right].$$
  
Since $\left(Y_{\tau_n}^ue^{\int_t^{\tau_n}h_sds}\right)_{n\geq1}$ 
 belongs to class $(\rm D)$,
we deduce
$$\lim_{n\to\infty}\mathbf{E}^{\mathbb{Q}}\left[Y_{\tau_n}^ue^{\int_t^{\tau_n}h_sds}\big|\mathscr{F}_t\right]=\mathbf{E}^{\mathbb{Q}}\left[Y_{T}^ue^{\int_t^{T}h_sds}\big|\mathscr{F}_t\right]=0.$$
On the other hand, $\left|Y_{\tau_n}e^{\int_t^{\tau_n}h_sds}\right|\leq(Y_\cdot^+)_*+(Y_\cdot^-)_*$. Moreover, from Fenchel inequality we have
\begin{align*}
\mathbf{E}^{\mathbb{Q}}\left[(Y_\cdot^+)_*\right]=\mathbf{E}\left[M_T(Y_\cdot^+)_*\right]\leq\mathbf{E}\left[e^{2p(Y_\cdot^+)_*}\right]+\frac{1}{2p}\mathbf{E}^{\mathbb{Q}}\left(\int_{0}^{T}|u_s|^2ds\right)-\frac{1+\ln2p}{2p}<+\infty.
\end{align*}
By the same reasoning, we can conclude $\mathbf{E}^{\mathbb{Q}}\left[(Y_\cdot^-)_*\right]<+\infty.$
It then follows from the dominated convergence theorem that
$$\lim_{n\to\infty}\mathbf{E}^{\mathbb{Q}}\left[Y_{\tau_n}e^{\int_t^{\tau_n}h_sds}|\mathscr{F}_t\right]=\mathbf{E}^{\mathbb{Q}}\left[Y_{T}e^{\int_t^{T}h_sds}|\mathscr{F}_t\right]=0.$$
Finally, we obtain
$$Y_t^u-Y_t\geq\lim_{n\to\infty}\mathbf{E}^{\mathbb{Q}}\left[(Y_{\tau_n}^u-Y_{\tau_n})e^{\int_t^{\tau_n}h_sds}|\mathscr{F}_t\right]=0.$$

Taking $u_s^*=H'_z(s,Y_s,Z_s)$, which implies that  $H(s,Y_s,\tilde{Z_s})-H(s,Y_s,Z_s)\geq (\tilde{Z_s}-Z_s)u_s^*$ for any $\tilde{Z_s}$. 
Then $$J(s,Y_s,u_s^*)=H(s,Y_s,Z_s)-Z_su_s^*,$$
from which we obtain $Y_t^{u^*}=Y_t$. 


The rest is to show $u^*\in\mathscr{A}$.  Since 
\begin{align*}
H(s,Y_s,Z_s)=&Z_su_s^*+J(s,Y_s,u_s^*)\\
\leq&Z_su_s^*-\frac{1}{2}\left((u_s^*)'+\frac{\gamma-1}{\gamma}\mu_s'\Sigma_s^{-1}\sigma_s\right)\left(u_s^*+\frac{\gamma-1}{\gamma}\sigma_s'\Sigma_s^{-1}\mu_s\right)+(\gamma-1) r_{s}^--\delta \theta\\
\leq&\left(Z_s-\frac{\gamma-1}{\gamma}\mu_s'\Sigma_s^{-1}\sigma_s\right)u_s^*-\frac{1}{2}|u_s^*|^2+(\gamma-1) r_{s}^--\delta \theta\\
\leq&\frac{1}{2}\left(2\left|Z_s-\frac{\gamma-1}{\gamma}\mu_s'\Sigma_s^{-1}\sigma_s\right|^2+\frac{1}{2}|u_s^*|^2\right)-\frac{1}{2}|u_s^*|^2+(\gamma-1) r_{s}^--\delta \theta,
\end{align*}
then it implies that 
\begin{align*}
\frac{1}{4}|u^*|^2\leq-H(s,Y_s,Z_s)+2|Z_s|^2+\frac{2(\gamma-1)^2}{\gamma^2}\mu_s'\Sigma_s^{-1}\mu_s+(\gamma-1) r_{s}^--\delta \theta,
\end{align*}
which indicates $\int_0^T|u_s^*|^2ds<\infty$, $\mathbb{P}$-a.s. 

The results that $\mathbf{E}^{\mathbb{Q}^*}\left[\int_{0}^{T}|u_s|^2ds\right]<+\infty$  and $(M_t^{u^*})_{t\in[0,T]}$ is a martingale follow the similar proof procedure in Theorem \ref{main}. In fact, define 
$$\chi_n:=\inf\left\{s\geq 0 :  \int_0^s|Z_r|^2+|u_r^*|^2dr>n\right\}\wedge T$$
and $\frac{d\mathbb{Q}^*_n}{d\mathbb{P}}:=M_{\chi_n}^{u^*}$. Recall that $J(s,Y_s,u_s^*)=H(s,Y_s,Z_s)-Z_su_s^*$, then
$$Y_0=Y_{\chi_n}+\int_{0}^{\chi_n}J(s,Y_s,u_s^*)ds-\int_{0}^{\chi_n}Z_sdW_s^{u^*}.$$
It follows from (\ref{Jestimate}) that
$$J(s,Y_s,u_s^*)\leq(\gamma-1)r^-_s-\delta\theta-\frac{\gamma+1}{4\gamma}|u_s^*|^2+\frac{(\gamma-1)|\gamma-2|}{4\gamma^2}\mu_s'\Sigma_s^{-1}\mu_s.$$
Then we obtain
\begin{align*}
Y_0=&\mathbf{E}^{\mathbb{Q}^*_n}\left[Y_{\chi_n}\right]+\mathbf{E}^{\mathbb{Q}^*_n}\left[\int_{0}^{\chi_n}J(s,Y_s,u_s^*)ds\right]\\
\leq&\mathbf{E}^{\mathbb{Q}^*_n}\left[Y_{\chi_n}^+\right]+\mathbf{E}^{\mathbb{Q}^*_n}\left[\int_{0}^{\chi_n}(\gamma-1)r^-_s-\delta\theta-\frac{\gamma+1}{4\gamma}|u_s^*|^2+\frac{(\gamma-1)|\gamma-2|}{4\gamma^2}\mu_s'\Sigma_s^{-1}\mu_sds\right].
\end{align*}
Noting that 
\begin{align}\label{transfer}
\mathbf{E}\left[M_{\chi_n}^{u^*}\ln M_{\chi_n}^{u^*}\right]=\mathbf{E}^{\mathbb{Q}^*_n}\left[\frac{1}{2}\int_{0}^{\chi_n}|u^*_s|^2ds\right].
\end{align}
Hence, by Fenchel's inequality, one can derive 
\begin{align*}
&\mathbf{E}^{\mathbb{Q}^*_n}\left[Y_{\chi_n}^+\right]\leq \mathbf{E}\left[e^{2p\left(Y_\cdot^+\right)_*}\right]+\frac{1}{4p}\mathbf{E}^{\mathbb{Q}^*_n}\left[\int_{0}^{\chi_n}|u^*_s|^2ds\right]-\frac{1+\ln 2p}{2p},\\
&\mathbf{E}^{\mathbb{Q}^*_n}\left[\int_0^{\chi_n}r_s^-ds\right]\leq \mathbf{E}\left[e^{2p\gamma\int_0^{\chi_n}r_s^-ds}\right]+\frac{1}{4p\gamma}\mathbf{E}^{\mathbb{Q}^*_n}\left[\int_{0}^{\chi_n}|u^*_s|^2ds\right]-\frac{1+\ln 2p\gamma}{2p\gamma},\notag\\
&\mathbf{E}^{\mathbb{Q}_n}\left[\int_0^{\chi_n}\mu_s'\Sigma_s^{-1}\mu_sds\right]\leq \mathbf{E}\left[e^{q\int_0^{\chi_n}\mu_s'\Sigma_s^{-1}\mu_sds}\right]+\frac{1}{2q}\mathbf{E}^{\mathbb{Q}^*_n}\left[\int_{0}^{\chi_n}|u^*_s|^2ds\right]-\frac{1+\ln q}{q}.
\end{align*}
Take suitable $p_0<p,q_0<q$ such that $\left(\frac{\gamma+1}{4\gamma}-\frac{1}{4p_0}-\frac{\gamma-1}{4p_0\gamma}+\frac{(\gamma-1)|\gamma-2|}{8q_0\gamma^2}\right)>0$ and we have
\begin{align*}
&\left(\frac{\gamma+1}{4\gamma}-\frac{1}{4p_0}-\frac{\gamma-1}{4p_0\gamma}+\frac{(\gamma-1)|\gamma-2|}{8q_0\gamma^2}\right)\mathbf{E}^{\mathbb{Q}^*_n}\left[\int_{0}^{\chi_n}|u^*_s|^2ds\right]\\
\leq&-Y_0+\mathbf{E}\left[e^{2p_0\left(Y_\cdot^+\right)_*}\right]+(\gamma-1)\mathbf{E}\left[e^{2p_0\gamma\int_0^{\chi_n}r_s^-ds}\right]+\frac{(\gamma-1)|\gamma-2|}{4\gamma^2}\mathbf{E}\left[e^{q_0\int_0^{\chi_n}\mu_s'\Sigma_s^{-1}\mu_sds}\right]\\
&+C_{p_0,q_0,\gamma}-\delta\theta,
\end{align*}
where $C_{p_0,q_0,\gamma}$ is a constant only depending on $p_0,q_0,\gamma$. Thanks to Fatou's lemma, by sending $n\to+\infty$, we obtain from (\ref{transfer}) that $$\mathbf{E}^{\mathbb{Q}^*}\left[\int_{0}^{T}|u_s|^2ds\right]=\mathbf{E}\left[M_T^{u^*}\ln M_T^{u^*}\right]<+\infty.$$
It follows from Jensen's inequality that $M^{u^*}\ln M^{u^*}$ is a $\mathbb{P}$-submartingale, which implies that $$\sup_{t\in[0,T]}\mathbf{E}\left[M^{u^*}_t\ln M^{u^*}_t\right]\leq\mathbf{E}\left[M^{u^*}_T\ln M^{u^*}_T\right]<+\infty.$$
Then de La Vall$\rm\acute{e}$e Poussin's lemma yields that $M^{u^*}$ is a martingale. 
\end{proof}

Finally if Assumption \ref{2} holds as well, it follows from taking the limit as $m$, $n$, and $k$ tend to infinity in (\ref{psi}) that $Y$ is bounded from above by a constant. $\hfill\square$

\subsection{The proof of Theorem \ref{main}}
The proof is divided into two steps.

\textbf{Step 1.} We demonstrate $Q$ is of class (D) under $\mathbb{P}$.
 
For each $n\geq 1$, define $$\tau_n:=\inf\left\{s\geq 0: \left[\int_0^s|Z_r|^2dr+\int_0^s|(1-\gamma)(\pi_r^*)'\sigma_r+Z_r|^2dr\right]\geq n\right\}\wedge T$$
and define a probability measure
$$\frac{d\mathbb{Q}_n}{d\mathbb{P}}:=Q_{\tau_n}=\mathscr{E}\left(\int\left[(1-\gamma)(\pi_s^*)'\sigma_s+Z_s\right]dW_s\right)_{\tau_n}.$$

Applying the following Fenchel inequality 
$$xy=\frac{x}{2p}(2py)\leq e^{2py}+\frac{x}{2p}(\ln x-\ln 2p-1), ~~\forall x>0, y\in\mathbb{R},$$
where $p>1$ is given in Assumption \ref{1},  to $Q_{\tau_n}Y_{\tau_n}$ gives that
\begin{align}\label{Lem7leq}
\mathbf{E}\left[Q_{\tau_n}Y_{\tau_n}\right]\leq\mathbf{E}e^{2p(Y_\cdot^+)_*}+\frac{1}{2p}\mathbf{E}\left[Q_{\tau_n}\ln Q_{\tau_n}\right]-\frac{1+\ln 2p}{2p}\mathbf{E}[Q_{\tau_n}].
\end{align}
Since $W^{\mathbb{Q}_n}=W-\int_0^\cdot(1-\gamma)\sigma_s'\pi_s^*+Z_s'ds$ is a $\mathbb{Q}_n$-Brownian motion,  where $\pi_s^*=\frac{1}{\gamma}\Sigma_s^{-1}\mu_s+\frac{1}{\gamma}\Sigma_s^{-1}\sigma_sZ_s'$ is given in (\ref{os}), then we rewrite BSDE \eqref{mar} as follows:
\begin{align}
Y_{\tau_n}=&Y_0-\int_0^{\tau_n}\bigg(-\frac{1}{2}Z_s\left(I_n+\frac{1-\gamma}{\gamma}\sigma_s'\Sigma_s^{-1}\sigma_s\right)Z_s'+\frac{\theta}{\psi}\delta^\psi e^{-\frac{\psi}{\theta}Y_s}+\frac{1-\gamma}{2\gamma}\mu_s'\Sigma_s^{-1}\mu_s\nonumber\\
&+(1-\gamma) r_{s}-\delta \theta\bigg) ds
+\int_0^{\tau_n}Z_sdW_s^{\mathbb{Q}_n}. \label{eq:bsde-qian}
\end{align}

Noting that the relative entropy between $\mathbb{Q}_n$ and $\mathbb{P}$ equals that
\begin{align*}
&\mathbf{E}\left[Q_{\tau_n}\ln Q_{\tau_n}\right]\\
=&\mathbf{E}^{\mathbb{Q}_n}\left[\frac{1}{2}\int_0^{\tau_n}\big|(1-\gamma)(\pi_s^*)'\sigma_s+Z_s)\big|^2ds\right]\\
=&\mathbf{E}^{\mathbb{Q}_n}\left[\frac{1}{2}\int_0^{\tau_n}Z_s\left(I_n+\frac{1-\gamma^2}{\gamma^2}\sigma_s'\Sigma_s^{-1}\sigma_s\right)Z_s'+\frac{2(1-\gamma)}{\gamma^2}\mu_s'\Sigma_s^{-1}\sigma_sZ_s'+\frac{(1-\gamma)^2}{\gamma^2}\mu_s'\Sigma_s^{-1}\mu_sds\right].
\end{align*}
Therefore, taking the expectation under $\mathbb{Q}_{n}$ on the both sides of \eqref{eq:bsde-qian} and using the above formulation of $\mathbf{E}\left[Q_{\tau_n}\ln Q_{\tau_n}\right]$, then we get that
\begin{align}\label{Lem7geq}
&\mathbf{E}\left[Q_{\tau_n}Y_{\tau_n}\right]\notag\\
=&\mathbf{E}\left[Q_{\tau_n}\ln Q_{\tau_n}\right]+\mathbf{E}^{\mathbb{Q}_n}\left[\int_0^{\tau_n}\frac{\gamma-1}{2\gamma^2}\left(Z_s\sigma_s'\Sigma_s^{-1}\sigma_s+2\mu_s'\Sigma_s^{-1}\sigma_s\right)^2ds\right]+\mathbf{E}^{\mathbb{Q}_n}[Y_0]\notag\\
&+\mathbf{E}^{\mathbb{Q}_n}\left[\int_0^{\tau_n}\frac{(1-\gamma)(\gamma+2)}{2\gamma^2}\mu_s'\Sigma_s^{-1}\mu_sds\right]+\mathbf{E}^{\mathbb{Q}_n}\left[\int_0^{\tau_n}\left(-\frac{\theta}{\psi}\delta^\psi e^{-\frac{\psi}{\theta}Y_s}-(1-\gamma) r_s+\delta \theta\right) ds\right]\notag\\
\geq&\mathbf{E}\left[Q_{\tau_n}\ln Q_{\tau_n}\right]+\mathbf{E}^{\mathbb{Q}_n}[Y_0]+\frac{(1-\gamma)(\gamma+2)}{2\gamma^2}\mathbf{E}^{\mathbb{Q}_n}\left[\int_0^{\tau_n}\mu_s'\Sigma_s^{-1}\mu_sds\right]\notag\\
&+(1-\gamma)\mathbf{E}^{\mathbb{Q}_n}\left[\int_0^{\tau_n}r_s^-ds\right]+\delta\theta T,
\end{align}where the last inequality 
uses $\gamma>1$ and $\theta<0$.

Now we use the Fenchel inequality to $Q_{\tau_n}\int_0^{\tau_n}r_s^-ds$ and $Q_{\tau_n}\int_0^{\tau_n}\mu_s'\Sigma_s^{-1}\mu_sds$ and derives that 
\begin{align*}
&\mathbf{E}^{\mathbb{Q}_n}\left[\int_0^{\tau_n}r_s^-ds\right]\leq \mathbf{E}\left[e^{2p\gamma\int_0^{\tau_n}r_s^-ds}\right]+\frac{1}{2p\gamma}\mathbf{E}\left[Q_{\tau_n}\ln Q_{\tau_n}\right]-\frac{1+\ln 2p\gamma}{2p\gamma},\notag\\
&\mathbf{E}^{\mathbb{Q}_n}\left[\int_0^{\tau_n}\mu_s'\Sigma_s^{-1}\mu_sds\right]\leq \mathbf{E}\left[e^{q\int_0^{\tau_n}\mu_s'\Sigma_s^{-1}\mu_sds}\right]+\frac{1}{q}\mathbf{E}\left[Q_{\tau_n}\ln Q_{\tau_n}\right]-\frac{1+\ln q}{q},
\end{align*}
where $p,q>1$ are given in Assumption \ref{1}. Plugging the above two inequalities into (\ref{Lem7geq}),  we have
\begin{align}\label{Lem7geq2}
&\mathbf{E}\left[Q_{\tau_n}Y_{\tau_n}\right]\notag\\
\geq&\left[1+\frac{(1-\gamma)(\gamma+2)}{2\gamma^2q}+\frac{1-\gamma}{2p\gamma}\right]\mathbf{E}\left[Q_{\tau_n}\ln Q_{\tau_n}\right]+\frac{(1-\gamma)(\gamma+2)}{2\gamma^2}\mathbf{E}\left[e^{q\int_0^{\tau_n}\mu_s'\Sigma_s^{-1}\mu_sds}\right]\notag\\
&+(1-\gamma)\mathbf{E}\left[e^{2p\gamma\int_0^{\tau_n}r_s^-ds}\right]+Y_0+\delta\theta T-\frac{1+\ln q}{q}\frac{(1-\gamma)(\gamma+2)}{2\gamma^2}+\frac{1+\ln 2p\gamma}{2p\gamma}(1-\gamma).
\end{align}
Combing (\ref{Lem7leq}) and (\ref{Lem7geq2}) together, we can conclude that
\begin{align}\label{marineq}
&\left[1-\frac{1}{2p}+\frac{1-\gamma}{2p\gamma}+\frac{(1-\gamma)(\gamma+2)}{2\gamma^2q}\right]\mathbf{E}\left[Q_{\tau_n}\ln Q_{\tau_n}\right]\notag\\
\leq&\mathbf{E}\left[e^{2p(Y_\cdot^+)_*}\right]-\frac{1+\ln 2p}{2p}+\frac{(\gamma-1)(\gamma+2)}{2\gamma^2}\mathbf{E}\left[e^{q\int_0^{\tau_n}\mu_s'\Sigma_s^{-1}\mu_sds}\right]+(\gamma-1)\mathbf{E}\left[e^{2p\gamma\int_0^{\tau_n}r_s^-ds}\right]-Y_0\notag\\
&-\delta\theta T+\frac{1+\ln q}{q}\frac{(1-\gamma)(\gamma+2)}{2\gamma^2}-\frac{1+\ln 2p\gamma}{2p\gamma}(1-\gamma)\notag\\
<&+\infty.
\end{align}
The choice of $p,q$ from Assumption \ref{1}, especially $q>\frac{p(\gamma-1)(\gamma+2)}{\gamma(1+2(p-1)\gamma)}$, guarantees that $1-\frac{1}{2p}+\frac{1-\gamma}{2p\gamma}+\frac{(1-\gamma)(\gamma+2)}{2\gamma^2q}>0$. 
Sending $n\to+\infty$, then we have $\mathbf{E}\left[Q_T\ln Q_T\right]<+\infty$.  On the other hand, Jensen's inequality yields $Q\ln Q$ is a $\mathbb{P}$-submartingale, which implies that $$\sup_{t\in[0,T]}\mathbf{E}\left[Q_t\ln Q_t\right]\leq\mathbf{E}\left[Q_T\ln Q_T\right]<+\infty.$$ It then follows from de La Vall$\rm\acute{e}$e Poussin's lemma that $Q$ is of class (D), and thereby it is a true martingale. 

\textbf{Step 2.} We show that $(c^*,\pi^*)\in\mathcal{A}$.

First we confirm that $\left(\mathcal{W}^*\right)^{1-\gamma}e^{Y}$ is of class (D).  It suggests from (\ref{up}) that
\begin{align}\label{UI}
\left(\mathcal{W}_t^*\right)^{1-\gamma}e^{Y_t}=\omega^{1-\gamma}e^{Y_0}\operatorname{exp}\left(\int_0^t\left(\theta\delta-\theta\delta^\psi e^{-\frac{\psi}{\theta}Y_s}\right)ds\right)Q_t,~~~~0\leq t\leq T.
\end{align}
Since $Y$ is bounded from above according to Proposition \ref{Y} and $Q$ is uniformly integrable, then $\left(\mathcal{W}^*\right)^{1-\gamma}e^{Y}$ is of class (D). 

Based on Proposition \ref{g}, $G^{c^*,\pi^*}$ is a local martingale. Taking a  localizing sequence $(\sigma_n)_{n\geq1}$ and on $\{t<\sigma_n\}$, we have
\begin{align*}
&\frac{(\mathcal{W}_t^*)^{1-\gamma}}{1-\gamma}e^{Y_t}+\delta\theta\mathbf{E}\left[\int_t^{T\wedge\sigma_n}\frac{(\mathcal{W}_s^*)^{1-\gamma}}{1-\gamma}e^{Y_s}ds\Big|\mathscr{F}_t\right]\\
=&\mathbf{E}\left[\frac{(\mathcal{W}_{T\wedge\sigma_n}^*)^{1-\gamma}}{1-\gamma}e^{Y_{T\wedge\sigma_n}}+\int_t^{T\wedge\sigma_n}\delta\frac{\left(c_s^*\right)^{1-\frac{1}{\psi}}}{1-\frac{1}{\psi}}\left((\mathcal{W}_s^*)^{1-\gamma}e^{Y_s}\right)^{1-\frac{1}{\theta}}ds\Big|\mathscr{F}_t\right].
\end{align*}
Since $\frac{(\mathcal{W}^*)^{1-\gamma}}{1-\gamma}e^{Y}\leq0$  and $\gamma,\psi>1$, the integrand of LHS is nonpositive and that of RHS is nonnegative. It then follows from the monotone convergence theorem with the class (D) property of $(\mathcal{W}^*)^{1-\gamma}e^{Y}$ that
\begin{align}\label{admi}
\frac{\left(\mathcal{W}_t^*\right)^{1-\gamma}}{1-\gamma}e^{Y_t}=\mathbf{E}\left[\int_t^Tf\left(c_s^*,\frac{\left(\mathcal{W}_s^*\right)^{1-\gamma}}{1-\gamma}e^{Y_s}\right)ds+\frac{\left(\mathcal{W}_T^*\right)^{1-\gamma}}{1-\gamma}\Bigg|\mathscr{F}_t\right].
\end{align}
This means that $V^{c^*,\pi^*}$ exists and $V^{c^*,\pi^*}<0$. Therefore, $(c^*,\pi^*)\in\mathcal{A}$. 
Finally,  (\ref{maxuti}) and (\ref{admi}) imply that  $\frac{\omega^{1-\gamma}}{1-\gamma}e^{Y_0}$ is the optimal value function. $\hfill\square$

\subsection{The proof of Theorem \ref{dual}}
For any $D\in\widehat{\mathcal{D}}$ and $y>0$,  define
\begin{align}\label{upo2}
R_t^{yD}:=\frac{\gamma (yD_t)^{\frac{\gamma-1}{\gamma}}}{1-\gamma}e^{\frac{Y_t}{\gamma}}+\int_0^tg\left(yD_s, \frac{ (yD_s)^{\frac{\gamma-1}{\gamma}}}{1-\gamma}e^{\frac{Y_s}{\gamma}}\right)ds,\quad t\in[0,T],
\end{align}
where $Y$ is defined by BSDE (\ref{mar}). Then
\begin{align}\label{up2}
dR_t^{yD}=&\frac{(yD_t)^{\frac{\gamma-1}{\gamma}}}{1-\gamma}e^{\frac{Y_t}{\gamma}}\Bigg[\left(Z_t+(\gamma-1)\xi_t\right)dW_t+\Big((1-\gamma)r_t-\delta\theta+\frac{1}{2\gamma}Z_tZ_t'+\theta\frac{\delta^\psi}{\psi}e^{-\frac{\psi}{\theta}Y_t}\notag\\
&+\frac{1-\gamma}{2\gamma}\xi_t\xi_t'-\frac{1-\gamma}{\gamma}\xi_t Z_t'-H(t,Y_t,Z_t)
\Big)dt\Bigg],
\end{align}
where $\xi$ comes from $dD_t=D_t(-r_tdt+\xi_tdW_t)$ and the generator of BSDE (\ref{mar}) has the following representation:
\begin{align*}
H(t,y,z)=&\frac{\theta}{\psi}\delta^\psi e^{-\frac{\psi}{\theta}y}
+(1-\gamma) r_{t}-\delta \theta+\frac{1}{2\gamma}zz'+\sup_{\mu_t+\sigma_t\eta'=0}\left(\frac{1-\gamma}{2\gamma}\eta\eta'-\frac{1-\gamma}{\gamma}\eta z'\right).
\end{align*}
Since $\gamma>1$, for any $D\in\widehat{\mathcal{D}}$ and $y>0$, then $R^{yD}$ is a local submartingale.  In particular, taking 
$$\xi_t^*=Z_t-\mu_t'\Sigma_t^{-1}\sigma_t-Z_t\sigma_t'\Sigma_t^{-1}\sigma_t,~~t\in[0,T],$$
such that $$\frac{dD_t^*}{D_t^*}=-r_tdt+\xi_t^*dW_t,~~t\in[0,T],$$
then $R^{yD^*}$ is a local martingale if $D^*\in \widehat{\mathcal{D}}$. Next, we claim $D^*\in \widehat{\mathcal{D}}$.   

Define 
$$U_t^{yD^*}:=\frac{\gamma}{1-\gamma}(yD_t^*)^{\frac{\gamma-1}{\gamma}}e^{\frac{Y_t}{\gamma}},~~t\in[0,T].$$Obviously $U^{yD^*}<0$, and we can verify that $U^{yD^*}$ satisfies \eqref{eq:dual-u} from \eqref{upo2}.  It suffices to show that $(D^*)^{\frac{\gamma-1}{\gamma}}e^{\frac{Y}{\gamma}}$ is of class (D). Indeed, from (\ref{up2}),  we  can obtain
\begin{align}\label{UI2}
(D_t^*)^{\frac{\gamma-1}{\gamma}}e^{\frac{Y_t}{\gamma}}=e^{\frac{Y_0}{\gamma}}\operatorname{exp}\left(\int_0^t\left(\frac{\theta\delta}{\gamma}-\frac{\theta\delta^\psi}{\gamma\psi} e^{-\frac{\psi}{\theta}Y_s}\right)ds\right)Q_t, ~~t\in[0,T].
\end{align}
Since $Y$ is bounded from above according to Proposition \ref{Y}, it follows from Theorem \ref{main} that (\ref{UI2}) is uniformly integrable.

Therefore, it implies that 
$U_0^{yD^*}=R_0^{yD^*}=\frac{\gamma y^{\frac{\gamma-1}{\gamma}}}{1-\gamma}e^{\frac{Y_0}{\gamma}}$.  
By the martingale optimal principle, one can derive that $$\max_{(c,\pi)\in\widehat{\mathcal{A}}}V_0^{c,\pi}=V_0^{c^*,\pi^*}=\frac{\omega^{1-\gamma}}{1-\gamma}e^{Y_0},$$ which uses $(c^*,\pi^*)\in\widehat{\mathcal{A}}$ by Theorem \ref{main}.  

Finally, choosing $y^*=\omega^{-\gamma}e^{Y_0}$, then, combining with the duality inequality, we get that  
$$V_0^{c^*,\pi^*}=\frac{\omega^{1-\gamma}}{1-\gamma}e^{Y_0}=U_0^{y^*D^*}+\omega y^*=\min_{y>0}\left(U_0^{yD^*}+\omega y\right)=\min_{y>0}\left(\min_{D\in\widehat{\mathcal{D}}}\left(U_0^{yD}+wy\right)\right).$$
$\hfill\square$

\section{Conclusion}
This paper studies the optimal consumption–investment problem for an investor with Epstein–Zin utility in a non-Markovian market.  The optimal consumption and investment strategies are characterized via a quadratic BSDE. The key observation is that the first component of the solution to this BSDE admits some given exponential moments when the market parameters are exponentially integrable. Consequently, we can verify that a certain exponential local martingale is a true martingale and further obtain the admissibility of the candidate optimal strategy. Then we introduce a dual problem for Epstein-Zin optimization and present the duality equality with non-Markovian markets. Several examples are illustrated and discussed. There remain many open questions, for example, do Assumptions \ref{1} and \ref{2} still hold for Vasicek's short rate model with unbounded interest rate? We leave this and other questions for future research.

\vspace{0.5cm}

\section*{Acknowledgment}
Dejian Tian was supported by the National Natural Science Foundation of China (No. 
12171471) and the Fundamental Research Funds for Central Universities (No. 2024KYJD2008). Harry Zheng was supported in part by Engineering and Physical Sciences Research Council of UK  (Grant No.  EP/V008331/1).  We are deeply grateful to the editor and the two anonymous referees for their constructive comments and suggestions that have helped to improve the paper from the previous version. We also greatly appreciate the presentations of the CUMT seminar on the works of \cite{delbaen2011uniqueness} and \cite{hu2018exponential} during June and July 2023, which inspired our ideas and significantly contributed to the development of our study. After our submission, upon recognizing the similarities but independent works between this work and \cite{hu2024utilityv4, hu2024utility}, we are also profoundly thankful to Professors Hu, Liang and Tang for their collegial communication and constructive comments and suggestions.  



\begin{thebibliography}{99}
\bibitem[Aurand and Huang(2023)]{aurand2023epstein}J. Aurand, Y.J. Huang, Epstein-Zin utility maximization on a random horizon, Mathematical Finance 33 (2023) 1370-1411.

\bibitem[Bansal and Yaron(2004)]{bansal2004risks}R. Bansal, A. Yaron, Risks for the long run: A potential resolution of asset pricing puzzles, J. Finance 59 (2004) 1481-1509.

\bibitem[Benzoni, Collin-Dufresne and Goldstein(2011)]{benzoni2011explaining}L. Benzoni, P. Collin-Dufresne, R. Goldstein, Explaining asset pricing puzzles associated with the 1987 market crash, J. Financ. Econ. 101 (2011) 552-573.


\bibitem[Briand et al.(2003)]{briand2003lp}P. Briand, B. Deylon, Y. Hu, E. Pardoux, L. Stoica, $L^p$ solutions of backward stochastic differential equations, Stochastic Process. Appl. 108 (2003) 109-129.

\bibitem[Briand and Hu(2006)]{briand2006bsde}P. Briand, Y. Hu, BSDE with quadratic growth and unbounded terminal value, Probab. Theory Relat. Fields 136 (2006) 604-618.


\bibitem[Chacko and Viceira(2005)]{chacko2005dynamic}G. Chacko, L.M. Viceira, Dynamic consumption and portfolio choice with stochastic volatility in incomplete markets, The Review of Financial Studies 18 (2005) 1369-1402. 

    
\bibitem[Cheridito and Hu(2011)]{cheridito2011optimal}P. Cheridito, Y. Hu, Optimal consumption and investment in incomplete markets with general constraints, Stoch. Dyn. 11 (2011) 283-299.    

\bibitem[Cox, Ingersoll and Ross(1985)]{cox1985theory}J.C. Cox, J.E. Ingersoll, S.A. Ross, A theory of the term structure of interest rates, Econometrica 53 (1985) 385-407.

\bibitem[Delbaen, Hu and Richou(2011)]{delbaen2011uniqueness}F. Delbaen, Y. Hu, A. Richou, On the uniqueness of solutions to quadratic BSDEs with convex generators and unbounded terminal conditions, Annales de l' Institut Henri Poincar$\acute{\rm e}$-Probabilit$\acute{\rm e}$s et Statistiques 47 (2011) 559-574.

\bibitem[Duffie and Epstein(1992)]{duffie1992stochastic}D. Duffie, L. Epstein, Stochastic differential utility, Econometrica 60 (1992) 353-394.

\bibitem[El Euch and Rosenbaum(2019)]{el2019characteristic}O. El Euch, M. Rosenbaum, The characteristic function of rough Heston models, Mathematical Finance 29 (2019) 3-38.


\bibitem[El Karoui, Peng and Quenez(2001)]{el2001dynamic}N. El Karoui, S. Peng, M.C. Quenez, A dynamic maximum principle for the optimization of recursive utilities under constraints, Annals of Applied Probability 11 (2001) 664-693.

\bibitem[Epstein and Zin(1989)]{epstein1989substitution}L. Epstein, S. Zin, Substitution, risk aversion, and the temporal behavior of consumption and asset returns: A theoretical framework, Econometrica 57 (1989) 937-969.


\bibitem[Fan, Hu and Tang(2020)]{fan2020uniqueness}S. Fan, Y. Hu, S. Tang, On the uniqueness of solutions to quadratic BSDEs with non-convex generators and unbounded terminal conditions, Comptes Rendus. Math{\'e}matique 358 (2020) 227-235.

\bibitem[Fan, Hu and Tang(2023)]{fan2023user}S. Fan, Y. Hu, S. Tang, A user’s guide to 1D nonlinear backward stochastic differential equations with applications and open problems, 2023, arXiv preprint, 2309.06233v1.

\bibitem[Feng and Tian(2023)]{feng2023optimal}Z. Feng, D. Tian, Optimal consumption and portfolio selection with Epstein-Zin utility under general constraints, Probability, Uncertainty and Quantitative Risk 8 (2023) 281-308.

\bibitem[Gatheral, Jaisson and Rosenbaum(2018)]{gatheral2018volatility}J. Gatheral, T. Jaisson, M. Rosenbaum, Volatility is rough, Quantitative finance 18 (2018) 933-949.


\bibitem[Gu, Lin and Yang(2016)]{gu2016dual}L. Gu, Y. Lin, J. Yang, On the dual problem of utility maximization in incomplete markets, Stochastic Process. Appl. 126 (2016) 1019-1035.



\bibitem[Herdegen, Hobson and Jerome(2022)]{herdegen2022infinite}M. Herdegen, D. Hobson and J. Jerome, The infinite-horizon investment--consumption problem for Epstein--Zin stochastic differential utility I: Foundations, Finance and Stochastics 27 (2022) 127-158.


\bibitem[Horst et al.(2014)]{horst2014forward}U. Horst, Y. Hu, P. Imkeller, A. R${\rm\acute{e}}$veillac, J. Zhang, Forward-backward systems for expected utility maximization, Stochastic Process. Appl. 124 (2014) 1813-1848.

\bibitem[Hu, Imkeller and M\"uller(2005)]{hu2005utility}Y. Hu, P. Imkeller, M. M\"uller, Utility maximization in incomplete markets, Ann. Appl. Probab. 15 (2005) 1691-1712.

\bibitem[Hu, Liang and Tang(2018)]{hu2018exponential} Y. Hu, G. Liang, S. Tang, Exponential utility maximization and indifference valuation with unbounded payoffs, 2018, arXiv preprint, 1707.00199v3.

\bibitem[Hu, Liang and Tang(2024a)]{hu2024utilityv4} Y. Hu, G. Liang, S. Tang, Utility maximization in constrained and unbounded financial markets: Applications to indifference valuation, regime switching, consumption and Epstein-Zin recursive utility, 2024a, arXiv preprint, 1707.00199v4.

\bibitem[Hu, Liang and Tang(2024b)]{hu2024utility} Y. Hu, G. Liang, S. Tang, Utility maximization in constrained and unbounded financial markets: Applications to indifference valuation, regime switching, consumption and Epstein-Zin recursive utility, 2024b, arXiv preprint, 1707.00199v5.



\bibitem[Karatzas, Lehoczky and Shreve(1987)]{karatzas1987optimal} I. Karatzas, J. P. Lehoczky, S. E. Shreve, Optimal portfolio and consumption decisions for a ``small investor'' on a finite horizon, SIAM J. Control Optim. 25 (1987) 1557-1586.

\bibitem[Kobylanski(2000)]{kobylanski2000backward}M. Kobylanski, Backward stochastic differential equations and partial differential equations with quadratic growth, Ann. Probab. 28 (2000) 558-602.

\bibitem[Kraft, Seiferling and Seifried(2017)]{kraft2017optimal}H. Kraft, T. Seiferling, F.T. Seifried, Optimal consumption and investment with Epstein-Zin recursive utility, Finance Stoch. 21 (2017) 187-226.

\bibitem[Kraft and Seifried(2014)]{kraft2014stochastic}H. Kraft, F.T. Seifried, Stochastic differential utility as the continuous-time limit of recursive utility, Journal of Economic Theory 151 (2014) 528-550.

\bibitem[Kraft, Seifried and Steffensen(2013)]{kraft2013consumption}H. Kraft, F.T. Seifried, M. Steffensen, Consumption-portfolio optimization with recursive utility in incomplete markets, Finance Stoch. 17 (2013) 161-196.

\bibitem[Kramkov and Schachermayer(1999)]{kramkov1999asymptotic}D. Kramkov, W. Schachermayer, The asymptotic elasticity of utility functions and optimal investment in incomplete markets, Ann. Appl. Probab. 9 (1999) 904-950.

\bibitem[Matoussi and Xing(2018)]{matoussi2018convex}A. Matoussi, H. Xing, Convex duality for Epstein-Zin stochastic differential utility, Math. Finance. 28 (2018) 991-1019.



\bibitem[Merton(1971)]{merton1971optimum}R.C. Merton, Optimum consumption and portfolio rules in a continuous-time model, J. Econom. Theory 3 (1971) 373-413.


\bibitem[Pu and Zhang(2024)]{pu2024consumption}J. Pu, Q. Zhang, Consumption and portfolio optimization with generalized stochastic differential utility in incomplete markets differential utility in incomplete markets, Systems \& Control Letters 183 (2024) 105680.

\bibitem[Schroder and Skiadas(1999)]{schroder1999optimal}M. Schroder, C. Skiadas, Optimal consumption and portfolio selection with stochastic differential utility, J. Econ. Theory 89 (1999) 68-126.

\bibitem[Schroder and Skiadas(2003)]{schroder2003optimal}M. Schroder, C. Skiadas, Optimal lifetime consumption-portfolio strategies under trading constraints and generalized recursive preferences, Stochastic Process. Appl. 108 (2003) 155-202.

\bibitem[Schroder and Skiadas(2005)]{schroder2005lifetime}M. Schroder, C. Skiadas, Lifetime consumption-portfolio choice under trading constraints, recursive preferences, and nontradeable income, Stochastic Process. Appl. 115 (2005) 1-30.

\bibitem[Seiferling and Seifried(2015)]{seiferling2015stochastic}T. Seiferling, F.T. Seifried, Stochastic differential utility with preference for information: existence, uniqueness, concavity, and utility gradients, 2015, working paper, http://ssrn.com/abstract=2625800.

\bibitem[Stein and Stein(1991)]{stein1991stock}E.M. Stein, J.C. Stein, Stock price distributions with stochastic volatility: an analytic approach, The Review of Financial Studies 4 (1991) 727-752.

\bibitem[Xing(2017)]{xing2017consumption}H. Xing, Consumption-investment optimization with Epstein-Zin utility in incomplete markets, Finance Stoch. 21 (2017) 227-262.

\bibitem[Yong(2004)]{yong2004some}J. Yong, Some estimates on exponentials of solutions to stochastic differential equations. J. Appl. Math. Stoch. Anal. 4 (2004) 287-316.

\end{thebibliography}
\end{document}